\documentclass[11pt,reqno]{article}
\usepackage{amsmath}
\usepackage{amsmath}
\usepackage{amsthm}
\usepackage{amssymb}
\usepackage{epsfig}
\usepackage{epstopdf}

\usepackage{calc}
\usepackage{xspace}
\usepackage{comment}
\usepackage{paralist}

\usepackage{verbatim}
\usepackage{url}
\usepackage{ulem}
\usepackage{setspace}
\usepackage{rotating}
\usepackage{amsmath,amssymb}
\usepackage{subfigure}
\usepackage{array,arydshln}	
\usepackage{multirow}
\usepackage{booktabs}
\usepackage{color,colortbl}
\usepackage{authblk}
\usepackage{hyperref}
\usepackage{enumitem}
\usepackage{tabularx}	
\usepackage[table]{xcolor}
\usepackage{adjustbox}
\usepackage{float}
\usepackage{makecell}
\usepackage{cite}
\usepackage{bm}
\usepackage{mathrsfs}
\usepackage{indentfirst}

\graphicspath{ {./pictures_of_paper/} }

\usepackage{natbib}
\DeclareMathOperator*{\argmax}{argmax}

\newtheorem{theorem}{Theorem}
\newtheorem{prop}{Proposition}[section]

\newtheorem{lm}{Lemma}

\newtheorem{@remark}{\bf Remark}
\newenvironment{remark}{\begin{@remark}\rm}{\end{@remark}}

\newcommand {\A} {\alpha}

\newcommand {\E} {\mathrm{E}}
\newcommand{\cov}{\mathrm{cov}}
\newcommand {\T} {\theta}

\newcommand {\ep} {\epsilon}
\newcommand {\pa} {\partial}
\newcommand {\pai} {\partial_{\T_{i_k}}}
\newcommand {\paj} {\partial_{\T_{j_k}}}

\newcommand {\pal} {\partial_{\T_l}}
\newcommand {\paa}{\partial^2_{\T\T'}}
\newcommand {\paaij}{\partial^2_{\T_{i_k} \T_{j_k}}}

\newcommand {\paail}{\partial^2_{\T_{i_k} \T_l}}
\newcommand {\paajl}{\partial^2_{\T_{j_k} \T_l}}

\newcommand {\paaa}{\partial^3_{\T_{i_k} \T_{j_k} \T_l}}

\newcommand{\tcr}{\textcolor{red}}
\newcommand{\tcb}{}

\newcommand {\les} {\lesssim}

\title{ Information matrix test for normality of innovations in stationary time series models}
\author{Zixuan Liu}
\affil{School of Mathematics and Information Science, Henan University of Economics and Law}
\author[2]{Junmo Song\thanks{Corresponding author}}
\affil{Department of Statistics, Kyungpook National University}

\begin{document}
\maketitle
\begin{abstract}
This study focuses on the problem of testing for normality of innovations in stationary time series models. 
To achieve this, we introduce an information matrix (IM) based test. While the IM test  was originally developed to test for model misspecification, our study addresses that the test can also be used to test for the normality of innovations in various time series models. We provide sufficient conditions under which the limiting null distribution of the test statistics exists. As applications, a first-order threshold moving average model, GARCH model and double autoregressive model are considered. We conduct simulations to evaluate the performance of the proposed test and compare with other tests, and provide a real data analysis. 
\\
\end{abstract}


\section{Introduction}\label{Sec:intro}

\setlength{\parindent}{1em}
Testing for normality has been an important issue in the statistical literature. Many useful tests and methods, such as  the Jarque-Bera (JB) test, the Shapiro-Wilk test, the Kolmogorov-Smirnov test and some graphical methods, have been proposed, and new techniques are constantly being developed. For the review on the normality test, we refer to \cite{yazici2007comparison} and  \cite{mbah2015shapiro}.

In the field of time series analysis, the normality of the errors is still an assumption worth checking. 
A simple way to test for the normality is, for example, to apply the Jarque-Bera or the Shapiro-Wilk tests developed in i.i.d. setting to the residuals obtained from a fitted time series model. Due to its convenience, such methods are usually used in practice but \tcb{one needs to pay attention when applying because the asymptotic distribution of the residual-based tests are not necessarily same as the one of the tests based on true i.i.d. errors} (cf. \cite{koul2006fitting}). Many works therefore have been devoted to showing the validity of each residual-based test. As examples for the JB test, \cite{kilian2000residual}, \cite{kulperger2005high}, \cite{yu2007high}, and \cite{lee:2012} showed the asymptotic validity of the residual-based JB test for vector error-correction models, GARCH models, ARMA models, and ARMA-GARCH models, respectively. 

We are also concerned with the normality test for the errors in time series models, but \tcb{  aim to propose a test applicable for more general time series models}. For this, we introduce an information matrix (IM) test derived from the information matrix equivalence, that is, the relation of the variance of the score function and the information matrix. \tcb{The IM test was originally proposed in order to check a tentative model is correctly specified} (see the original paper by \cite{white:1982} for details).  Subsequently, the IM test has been developed mostly for testing misspecification in various models. See, for example,  \cite{hall1987information}, \cite{reynolds1991testing}, \cite{zhang2001information}, and \cite{abad2010testing}. In time series analysis, \cite{weiss:1984} and \cite{furno:1996} used the IM test for the specification of ARMA-ARCH model and the regression model with ARMA errors, respectively.  

In this study, \tcb{ we shall address that  the IM test is also applicable and practically useful for testing the normality of the errors}. To understand the basic idea, we need to recall the fact that the IM equivalence holds when the model considered is correctly specified and, at the same time, the distribution of error is exactly identified. The previous studies utilizing the IM test for model misspecification testing assumed that the distribution of innovation is known. Conversely to this, if the model is properly specified, it means that \tcb{the IM equivalence could also be used to identify the error distribution.} Particularly given that \tcb{the Gaussian quasi-MLE (QMLE) is widely used in various time series models}, it would be worth exploring the applicability of the IM test for the normality test for the errors because \tcb{the IM based normality test can be performed immediately after obtaining the Gaussian QML estimate}, similar to other likelihood based tests such as the score test and the likelihood ratio test. However, to the best of our knowledge, little work has been made on this topic and thus \tcb{we investigate the IM test and its properties as a tool for testing the normality of innovations in general time series models.} 

The organization of this paper is as follows. In Section \ref{Sec:2}, we construct the IM test statistic and provide sufficient conditions, under which the IM test statistics converges in distribution to a chi-square distribution under the null hypothesis. In Section \ref{Sec:3}, we apply the introduced test to the first-order threshold moving average model, the GARCH model, and the double AR model. We compare the performance of the IM test with other residual based tests  through simulations in section \ref{Sec:simu}. A real data example is provided in section \ref{Sec:real} and the technical proofs are given in Appendix. 

\section{Main results}\label{Sec:2}
Let $\left\{ X_{t}|t \in \mathbb{Z} \right\}$  be a strictly stationary and ergodic time series and  assume that the series can be modelled by
\begin{equation}\label{model1}
	X_{t}=\mu_{t}(\theta)+ \sigma_{t}(\theta) e_{t},
\end{equation}
where $\mu_t(\theta)$  and  $\sigma_t(\theta)$ are measurable functions of $\{X_{t-1}, X_{t-2}, \cdots \}$ with the parameter $\theta \in \mathbb{R}^p$  and   $\left\{ e_{t}|t \in \mathbb{Z} \right\}$ is a sequence of i.i.d random variables  with zero mean and unit variance. We assume that the parameter space $\Theta$ is a compact subset of $\mathbb{R}^p$ and the true parameter $\theta_0$ is in the interior of $\Theta$.  Various time series models such as the traditional ARMA models,  GARCH-type models,  and ARMA-GARCH models can be represented by the model (\ref{model1}).

To estimate the model above, consider the Gaussian QMLE given as
\[ \hat\theta_n = \argmax_{\theta \in \Theta}  \sum_{t=1}^n l(X_t; \theta),\]
where
\begin{eqnarray}\label{OF}
 l(X_t;\theta) =  -\frac{1}{2}\Big( \log \sigma_t^2(\theta) +\frac{(X_t -\mu_t(\theta))^2}{\sigma_t^2(\theta)} \Big).
\end{eqnarray}
For the models where  $\{\mu_t(\theta)|1\leq t \leq n\}$ and $\{\sigma_t^2(\theta)| 1\leq t\leq n\}$ include some unobservable terms due to the initial value issue, it needs to replace the processes with other feasible ones, which can usually be obtained via recursion. In such cases, we denote the approximated processes by $\{ \tilde \mu_t(\theta) | 1\leq t \leq n\}$ and $\{ \tilde \sigma^2_t(\theta) | 1\leq t \leq n\}$, respectively, and the QMLE is then redefined as follows:
\begin{eqnarray}\label{QMLE}
\hat\theta_n = \argmax_{\theta \in \Theta}  \sum_{t=1}^n \tilde l(X_t; \theta),
\end{eqnarray}
where $\tilde l(X_t;\theta)$ is the one obtained from $l(X_t;\theta)$ by replacing $\mu_t(\theta)$ and $\sigma_t^2(\theta)$ with $\tilde \mu_t(\theta)$ and $\tilde \sigma_t^2(\theta)$, respectively.
See, for example, \cite{berkes:2003} and \cite{francq:zakoian:2004} for GARCH models and ARMA-GARCH models, respectively. Hereafter, we mainly state in terms of the estimator (\ref{QMLE}) because $\tilde l(X_t;\theta)$ can be regarded as  $l(X_t;\theta)$ in the case that  $\mu_t(\theta)$ and $\sigma_t^2(\theta)$ are known. Concerning notations, we use $\pa_\T$ and $\pa^2_{\T\T'}$ to denote $\pa/\pa \T$ and $\pa^2/\pa \T \pa \T'$, respectively, and $\| \cdot\|$ denotes any vector or matrix norm. 

The consistency and the asymptotic normality of the QMLE are essential in deriving the limiting null distribution of the test statistics which will be introduced below. For this, we consider the following conditions:
\begin{enumerate}	
	\item[\bf{A1.}]$\left\{ l(X_{t};\theta) | t \in \mathbb{Z} \right\}$ is strictly stationary and ergodic for each $\theta \in \Theta$.
	\item[\bf{A2.}] The true parameter $\T_0$ is identifiable. That is, $\E[ l(X_t;\theta)]$ has a unique maximum at $\T_0$.
	\item[\bf{A3.}] (a)
	$\displaystyle \E \sup_{\T\in\Theta} \big|l(X_t;\T)\big|<\infty\quad\mbox{and (b)}\quad\frac{1}{n}\sum_{t=1}^n \sup_{\T\in\Theta} \big| l(X_t;\theta)-\tilde l(X_t;\theta)\big|=o(1)\quad a.s.$
\item[\bf{A4.}] $\mu_t(\T)$ and $\sigma^2_t(\T)$ are   twice  continuously differentiable with respect to $\T$.
	\item[\bf{A5.}]  $ \E\left[ \partial_{\theta}l(X_{t};\theta_{0})\partial_{\theta'}l(X_{t};\theta_{0}) \right]$ exists and $\E\left[ \partial^{2}_{\theta \theta'} l(X_{t};\theta_{0}) \right]$ is positive definite.
\item[\bf A6.] For some neighborhood $N(\T_0)$ of $\T_0$,
 $$\displaystyle\E \sup_{\theta \in N(\T_0)}\big\|\paa  l(X_t;\T)\big\|  <\infty.$$
\item[\bf A7.]$\displaystyle \frac{1}{\sqrt{n}}\sum_{t=1}^n \big\| \pa_{\theta}\,l(X_t;\theta_0)-\pa_{\theta}\,\tilde l(X_t;\theta_0)\big\|= o(1)\quad a.s.$

\item[\bf A8.]For some neighborhood $N(\T_0)$ of $\T_0$, $$\displaystyle \frac{1}{n}\sum_{t=1}^{n}\sup_{\theta \in N(\T_0)}\big\|\paa  l(X_t;\theta)
-\paa  \tilde l(X_t;\theta)\big\|= o(1)\quad a.s.$$
\end{enumerate}
 {\bf A1} is usually guaranteed by the stationarity and ergodicity of $\{X_t\}$. {\bf A2} and {\bf A3} are the standard assumptions to show the consistency of the estimator. {\bf A6} and the continuity of $\paa l(X_t;\T)$ assured by {\bf A4} indeed yield $\E \sup_{\theta \in N_2(\T_0)}\big\|\paa  l(X_t;\T)-\paa  l(X_t;\T_0)\big\|  <\infty$, from which together with  {\bf A7} and {\bf A8} one can derive asymptotic normality. In the case of $\ep_t\sim N(0,1)$, the positive definiteness of $\E\left[ \partial^{2}_{\theta \theta'} l(X_{t};\theta_{0}) \right]$  can be readily shown just by checking that $z'\pa_\T \sigma_t^2(\T_0)$ and  $z'\pa_\T \mu_t(\T_0)$ are equal to zero almost surely only for $z={\bf 0}$ (cf. Lemma \ref{nonsing} below). Under the assumptions above, one can obtain the following asymptotic result. 

\begin{theorem}\label{thm_qmle}
	Suppose that the assumptions $\bf{A1}$--$\bf{A3}$ hold. Then, $\hat\theta_n$ converges almost surely to $\T_0$. If additionally the assumptions {\bf A4}--{\bf A8} hold and $\T_0$ is in the interior of $\T$, we have
	\[ \sqrt{n}(\hat\T_{n}-\T_0) \stackrel{d}{\longrightarrow} N\big(0, \mathcal{J}^{-1}\mathcal{I}\mathcal{J}^{-1} \big),\]
	where $\mathcal{J}= \E\left[ \partial^{2}_{\theta \theta'} l(X_{t};\theta_{0}) \right]$ and $\mathcal{I}= \E\left[ \partial_{\theta}l(X_{t};\theta_{0})\partial_{\theta'}l(X_{t};\theta_{0}) \right]$.
\end{theorem}

Now, we consider the problem of testing for the normality of the error distribution. That is, the hypotheses of our interest are :
\begin{eqnarray*}
H_0 : e_t \mbox{ follows a normal distribution.}\quad v.s. \quad H_1:  \mbox{ not } H_0.
\end{eqnarray*}
To this end, we employ the IM test, which comes from the information matrix equivalence, i.e., the equivalence of the variance of the score function and the information matrix. Under $H_0$, this equivalence can also be shown for the time series model of (\ref{model1}), that is,
\[ \E\left[ \partial_{\theta}l(X_{t};\theta_{0})\partial_{\theta'}l(X_{t};\theta_{0})\right]+\E\left[ \partial^{2}_{\theta \theta'} l(X_{t};\theta_{0})\right]=0 \]
(see Lemma \ref{lm1}), from which we can consider the following statistics
\begin{eqnarray}\label{elm}
\bigg\{ \sum_{t=1}^n\partial_{\theta_i}\tilde l(X_{t};\hat\theta_n)\partial_{\theta_j}\tilde l(X_{t};\hat\theta_n) +  \sum_{t=1}^n\partial^{2}_{\theta_i \theta_j} \tilde l(X_{t};\hat\theta_n) \bigg| 1\leq i,j \leq p \bigg\},
\end{eqnarray}
where $\theta_i$ and $\T_j$ are the $i$th and $j$th elements in $\T$, respectively.
Here, we note that, as explained in \cite{white:1982}, it may be inappropriate to construct a test statistics using all the elements in (\ref{elm}) because some elements can be zero or a linear combination of others. In this regard, we construct a statistics using some of the elements. 

Let $\tilde d_k(X_t;\T)=\pa^2_{\T_{i_k} \T_{j_k}} \tilde l ( X_{t};\theta ) + \pa_{\T_{i_k}} \tilde l ( X_{t};\theta ) \pa_{\T_{j_k}} \tilde l ( X_{t};\theta ) $ for some $i_k \leq  j_k \leq p$ and define $q$-dimensional vector  $\tilde d(X_t;\T)= (\tilde d_1(X_t;\T),\cdots, \tilde d_q(X_t;\T))'$, where $q\leq p(p+1)/2$.  Then, the IM test statistics is given as
 \begin{eqnarray*}
T_n:=\frac{1}{\sqrt{n}} \sum_{t=1}^n \tilde d(X_t;\hat \T_n).
\end{eqnarray*}
To derive the limiting null distribution, we define some notations. $d(X_t;\T)$ is the counterpart of $\tilde d(X_t;\T)$ obtained by replacing $\tilde l(X_t;\T)$ with $l(X_t;\T)$  and $\nabla d(X_t;\T_0) $ is a matrix whose $k,l$ element is given by  $\partial_{\theta_l} d_k(X_t;\T_0)$, where $k \leq q$ and $l \leq p$.  To get the limiting distribution of $T_n$, further following  conditions are required.
\begin{enumerate}	
\item[\bf{C1.}] $l(x;\T)$ is three times differentiable with respect to $\T$ and is continuous in $\T$ for each $x$.
\item[\bf{C2.}]For some neighborhood $N(\theta_{0})$ of $\theta_{0}$, 
\[ \E \sup_{\T\in N(\T_0)} \big\| \nabla d(X_t;\T) \big\|<\infty. \]
\item[\bf{C3.}] $\cov(d(X_t;\T_0))$ and  $\cov \left(d(X_t;\T_0)- \E[\nabla d(X_t;\T_0)] \mathcal{J}^{-1} \pa_\T l(X_t;\T_0) \right)$ exist.
\item[\bf{C4.}]For some neighborhood $N(\theta_{0})$ of $\theta_{0}$,
\[\frac{1}{\sqrt{n}} \sum_{t=1}^{n} \sup_{\theta \in N(\theta_{0})}  \big\| \partial^{2}_{\theta \theta'} l\left( X_{t};\theta \right) -\partial^{2}_{\theta \theta'} \tilde l\left( X_{t};\theta \right) \big\| = o_P(1) \]
and
\[ \frac{1}{\sqrt{n}} \sum_{t=1}^{n} \sup_{\theta \in N(\theta_{0})}  \big\|  \partial_{\theta} l(X_{t};\T) \partial_{\theta^{'}} l(X_{t};\T) - \partial_{\theta } \tilde{l}(X_{t};\T) \partial_{\theta^{'} } \tilde{l}(X_{t};\T)\big\| = o_P(1). \]
\end{enumerate}
Below our main theorem, condition {\bf A8} is replaced with condition {\bf C4}.

\begin{theorem}\label{thm_main}
	Suppose that the assumptions {\bf A1}--{\bf A7}  and the conditions {\bf C1}--{\bf C4} hold. Then, under the null hypothesis,
\begin{eqnarray*}
 T_n\stackrel{d}{\longrightarrow} N_q ( {\bf 0},V(\T_0) ),
\end{eqnarray*}
where $V(\T_0)=\cov \left(d(X_t;\T_0)- \E[\nabla d(X_t;\T_0)] \mathcal{J}^{-1} \pa_\T l(X_t;\T_0) \right)$.
\end{theorem}

\begin{remark}
In cases that  the strong consistency and asymptotic normality of the model considered are already established, it just needs to check that the conditions {\bf C1}--{\bf C4} are met.
\end{remark}

\begin{theorem}\label{thm_test}
Let $\hat V_n$ be a consistent estimator of $V(\T_0)$. If the assumptions in Theorem \ref{thm_main}  hold and $V(\T_0)$ is nonsingular, we have 
\[ T_n' \hat V_n^{-1} T_n \stackrel{d}{\longrightarrow} \chi^2_q.\]
Hence, we reject $H_0$ if $T_n' \hat V_n^{-1} T_n > C_\A$, where $C_\A$ is the $(1-\A)$-quantile of $\chi^2_q$.
\end{theorem}

\begin{remark}\label{VT0}
One can avoid calculating the third derivatives included in  $\nabla d(X_t;\T_0)$. Using Lemma \ref{lm2}, we can see that under $H_0$,
\[V(\theta_{0}) = \E \left[ d(X_t;\T_0) d(X_t;\T_0)' \right] + \E \left[ d(X_t;\T_0) \pa_{\T'} l(X_t;\theta_{0}) \right]\mathcal{I}^{-1} \E \left[ \pa_\T l(X_t;\theta_{0}) d(X_t;\T_0)' \right].\]
Hence, a natural estimator for $V(\T_0)$ is given by
\begin{eqnarray*}
\hat V_n &=& \frac{1}{n} \sum_{t=1}^n \tilde d(X_t;\hat\T_n) \tilde d(X_t;\hat\T_n)' \\
&&+ \frac{1}{n}\sum_{t=1}^n \tilde d(X_t;\hat\T_n) \pa_{\T'} \tilde l(X_t;\hat\T_n) \Big(\sum_{t=1}^n \pa_\T \tilde l(X_t;\hat\T_n) \pa_{\T'} \tilde l(X_t;\hat\T_n) \Big)^{-1} \sum_{t=1}^n \pa_\T \tilde l(X_t;\hat\T_n) \tilde d(X_t;\hat\T_n)'.
\end{eqnarray*}
\end{remark}

\begin{remark}{\label{rm_elt}}
Selecting an optimal set of the elements in (\ref{elm}) for $\tilde d(X_t;\T)$ is indeed a  practical issue. 
Unfortunately, finding an objective criterion for selecting an optimal set that yields a best performance seems not easy. This represents a weakness of the IM test, and a preliminary simulation may be required to search for an optimal set of the elements. Implementing the test using only the diagonal elements in the matrix (\ref{elm}), i.e., $i=j$, could be a practical choice as all parameters are considered equally weighted. Based on our simulation study, the test with such diagonal elements does not always produce best result, but it often performs above average compared to the test using other combinations of the elements. 
\end{remark}

\begin{remark}{\label{rm_merit}}
In the simulation study below, we assess the performance of the IM test comparing with several normality tests developed for i.i.d data, such as the Jarque-Bera test, the Kolmogorov–Smirnov, and the Anderson–Darling test. We conduct the existing tests using residuals although the limiting null distributions of most of these tests have not been established for the time series models considered in the simulations. Based on our findings, the IM test outperforms particularly in cases where the error distribution is not heavy-tailed and also shows a satisfactory performance in other heavy-tailed cases. For more details, see section \ref{Sec:simu} below.  
\end{remark}
\section{Applications}\label{Sec:3}
In this section, we provide three applications of the IM test to the threshold MA(1) (TMA(1)) model, the GARCH model, and the first-order double AR (DAR(1)) model. In the first application, we will examine all conditions necessary to establish the asymptotic properties of the MLE and to obtain the limiting null distribution of the IM test. This serves as an example for models where the strong consistency and the asymptotic normality of the MLE are not established.  In the second and third  applications, the asymptotics of the QMLE  have been well investigated. For example, see \cite{francq:zakoian:2004} and \cite{ling:2004} for the GARCH model and DAR(1) model, respectively. Hence, our focus in these applications is to check whether the conditions {\bf C1} - {\bf C4} hold for these models. Since our objective is to test the normality of the error distribution, we derive the following results under $H_0$.

\subsection{Threshold  MA(1) model}\label{Sec:3_1}
Consider the following TMA(1) model:
\begin{eqnarray*}
	X_{t}= \left(\phi + \xi I (X_{t-1} \le u) \right) \sigma e_{t-1} + \sigma e_{t}
\end{eqnarray*}
where $I(\cdot)$ is the indicator function and $\{e_t\}$ is a sequence of i.i.d. random variables with mean zero and unit variance.  The threshold value $u \in \mathbb{R}$  is assumed to be fixed in prior.  We denote the parameter vector by $\theta=(\phi,\xi,\sigma^2)' \in \Theta \subset \mathbb{R}^2\times(0,\infty)$.  Noting that 
 $\E (X_t |\mathcal{F}_{t-1})=\left(\phi + \xi I (X_{t-1} \le u) \right) \sigma e_{t-1}$ and $\mathrm{var}(X_t|\mathcal{F}_{t-1})=\sigma^2$,  one can define the QMLE for the TMA(1) model as follows:
 \begin{eqnarray}\label{mle_TMA}
     \hat\theta_n = \argmax_{\theta \in \Theta}  \sum_{t=1}^n \tilde l_t(\theta),
 \end{eqnarray}
 where 
 \[ \tilde l_t(\theta) = -\frac{1}{2} \log \sigma^2-\frac{1}{2\sigma^2} \left\{X_t-\left(\phi + \xi I (X_{t-1} \le u) \right) \sigma \tilde \ep_{t-1}\right\}^2\]
 and $\{\tilde \ep_t | 1\leq t \leq n\}$ is the approximated process for $\{e_t\}$ given recursively by
 \begin{eqnarray}\label{tma.tep}
 \tilde \ep_t(\T):=\tilde \ep_t = \frac{1}{\sigma}X_t-(\phi+\xi I(X_{t-1} \le u)) \tilde\ep_{t-1}
 \end{eqnarray}
 with $\tilde \ep_0 = 0$ as the initial value. 
For the stationarity and ergodicity, we assume that for some constants $c_1<1$, $c_2$, and $c_3$,
\begin{eqnarray}\label{tma.ps}
   \Theta=\{ \T\, |\,  \left| \phi \right| \leq c_1, \left| \phi + \xi \right| \leq c_1, 0<c_2\leq \sigma^2 \leq c_3\}
\end{eqnarray}
and that the true parameter $\T_0$ lies in the interior of $\Theta$. Then, by \cite{ling:2007}, the process $\{X_t | t\in \mathbb{Z} \}$ of TMA(1) model is strictly stationary, ergodic, and further invertible.  
 Hereafter in this subsection, $l_t(\T)$ denotes the counterpart of $\tilde l_t(\T)$ obtained by substituting $\tilde \ep_{t-1}$ with $\ep_{t-1}$ defined as the solution of
 \begin{eqnarray}\label{tma.ep}
 \ep_t(\T):= \ep_t = \frac{1}{\sigma}X_t-(\phi+\xi I(X_{t-1} \le u)) \ep_{t-1}\quad\mbox{for}\  t\in\mathbb{Z}.
 \end{eqnarray}
Thanks to Theorem A.2 of \cite{ling:2005}, the process $\{\ep_t\}$ is well defined, and one can see that it is also strictly stationary and ergodic due to the stationarity and ergodicity of $\{X_t\}$. We also note that $\ep_t(\T_0)= e_t$. 
\begin{remark}
According to \cite{ling:2007}, the condition $|\xi|\sup_{x} |x f(x)|<1$ is further required for the process to be stationary and ergodic, where $f$ is the density of the error distribution. Since we assume that $e_t \sim N(0,1)$ under $H_0$, one can check that this condition is fulfilled for all $\T \in \Theta$.
\end{remark}
We now check the conditions introduced in Section \ref{Sec:2}. {\bf A1} directly follows from the stationarity and ergodicity of $\{X_t |t\in\mathbb{Z}\}$. To deal with  {\bf A2}, let $h_t(\T):=(\phi+\xi I (X_t \leq u) )\sigma$ and note that
\begin{eqnarray*}
	\E \big[X_t-h_{t-1}(\T)\ep_{t-1}\big]^2&=& \E \big[\sigma_0 e_t +h_{t-1}(\T_0)e_{t-1}-h_{t-1}(\T)\ep_{t-1}\big]^2\\
	&=&
	\sigma_0^2 +\E \big[h_{t-1}(\T_0)e_{t-1}-h_{t-1}(\T)\ep_{t-1} \big]^2
\end{eqnarray*}
Then, we have
\begin{eqnarray*}
	\E\,  l_t(\T) 
	&=& -\frac{1}{2} \Big( \log \sigma^2 +  \frac{\sigma^2_0}{\sigma^2} \Big) - \frac{1}{2\sigma^2} \E \big[ h_{t-1}(\T_0)e_{t-1}-h_{t-1}(\T)\ep_{t-1}\big]^2.
\end{eqnarray*}
One can readily see that the first term on the left side of the above equality is maximized at $\sigma^2=\sigma_0^2$. Since the second term is non-positive, $\E\, l_t(\T)  $ reaches the maximum when $h_{t-1}(\T_0)e_{t-1}-h_{t-1}(\T)\ep_{t-1}$is equal to zero almost surely, which implies that $\E\, l_t(\T) $ is maximized at $\T=\T_0$.
Next, since $l_t(\T)= -\frac{1}{2}\log\sigma^2 -\frac{1}{2}\ep^2_t$, it follows from the boundedness of $\Theta$ and Lemma \ref{lm_TMA0} that
\begin{eqnarray*}
	\E \sup_{\T\in\Theta} \left| l_t(\T) \right| &\les &  1+\E \sup_{\T\in\Theta} \ep^2_{t}<\infty.
\end{eqnarray*}
Here, we used the relation $A_n\les B_n$, where $A_n$ and $B_n$ are nonnegative, to denote that $A_n \leq KB_n$ for a positive constant $K$, and we shall use throughout the paper. {\bf A3}(b) can be shown by using Lemmas \ref{lm_TMA0} and \ref{lm_TMA1} as follows:
\begin{eqnarray*}
   \frac{1}{n} \sum_{t=1}^{n} \sup_{\T\in\Theta} \left| l_t(\T) - \tilde{l}_t(\T) \right|
	&=& \frac{1}{n} \sum_{t=1}^{n} \sup_{\T\in\Theta} \frac{1}{2}\left| \tilde \ep_{t}^2-\ep^2_{t} \right| \\
		&\les& \frac{1}{n} \sum_{t=1}^{\infty}\rho^t\big(1+\sup_{\T\in\Theta}|\ep_t|\big)=O\big(\frac{1}{n}\big)\quad a.s.
\end{eqnarray*}
Noting that $\mu_t(\T)=(\phi+\xi I(X_{t-1}\leq u))\sigma \ep_{t-1}$ and $\sigma_t(\T)=\sigma^2$ and that $\ep_t$ can be expressed as in (\ref{ept}), one can see that {\bf A4} is satisfied.  {\bf A5} and {\bf A6} are shown in Lemma \ref{nonsing} under $H_0$ and Lemma \ref{moment2}, respectively. {\bf A7} and {\bf A8} follows from Lemma \ref{lm_approx}. Therefore, the MLE defined in (\ref{mle_TMA}) is strongly consistent and satisfy the asymptotic normality. Furthermore, one can readily show {\bf C2} by using the results in Lemma \ref{moment2}. {\bf C3} can be also shown by using Lemma \ref{moment2} and Lemma \ref{lm2} (cf. see Remark \ref{VT0}). {\bf C4} comes from Lemma \ref{lm_approx}. Hence, we have the following result for TMA(1) model. 
\begin{theorem}\label{thm_TMA}
	Under $H_0$, it holds that
\begin{eqnarray*}
 T_n\stackrel{d}{\longrightarrow} N_q ( {\bf 0},V(\T_0) ),
\end{eqnarray*}
where $V(\T_0)$ is the one given in Theorem \ref{thm_main}. Thus, if $\hat V_n$ is a consistent estimator of $V(\T_0)$, we have that
\[ T_n' \hat V_n^{-1} T_n \stackrel{d}{\longrightarrow} \chi^2_q.\]
\end{theorem}

\subsection{GARCH model}\label{Sec:3_2}
\noindent Consider the  following GARCH(p,q) models:
\begin{equation}\label{GARCH}
	\begin{aligned}
		X_t &=\sigma_{t} e_t \\
		\sigma^{2}_{t}&=\omega+\sum_{i=1}^{p} \alpha_{i} X^{2}_{t-i}+\sum_{j=1}^{q} \beta_{j}\sigma^{2}_{t-j}
	\end{aligned}	
\end{equation}

\noindent where $\omega >0$, $\alpha_{i} \ge 0$, $\beta_{j} \ge 0$ and $\{ e_{t}| t \in \mathbb{Z} \}$ is a sequence of i.i.d random variables with zero mean and unit variance. The parameter vector is denoted by $\T=(\omega,\alpha_{1},\cdots,\alpha_{p},\beta_{1},\cdots,\beta_{q})^{'} \in \Theta \subset { (0,\infty) \times [0,\infty)^{p+q}}$ and the true parameter that generates the process $\{X_t\}$ is denoted by $\T_0$. We assume that $\{X_t\}$ is strictly stationary and ergodic. The detailed conditions for the GARCH model to have such solution can be found, for example, in \cite{bougerol:1992} and \cite{chen:an:1998}.

As an estimator for $\T$, we employ the QMLE of 
\cite{francq:zakoian:2004} given by
\begin{equation} \label{qmle-garch}
	\begin{aligned}
		\hat{\theta}_{n} &= \argmax_{\theta \in \Theta} \sum_{t=1}^{n} \tilde{l}_{t}(\theta)
	\end{aligned}
\end{equation}
where 
\[\tilde{l}_{t}(\theta)=-\frac{1}{2} \left( \log \tilde{\sigma}^{2}_{t}(\theta)+\frac{X^2_t}{\tilde{\sigma}^{2}_{t}(\theta)}  \right)\]
and 
$\{\tilde{\sigma}^{2}_{t}| 1 \le t \le n \}$ is the processes defined recursively by
\begin{eqnarray*}
	          \tilde{\sigma}^{2}_{t}(\T): = \tilde{\sigma}^{2}= \omega+\sum_{i=1}^{p} \alpha_{i} X^2_{t-i} + \sum_{j=1}^{q} \beta_{j} \tilde{\sigma}^{2}_{t-j}.
			\end{eqnarray*}
Here the initial values are assumed to be given properly. $l_t(\T)$, the stationary version of $\tilde l_t(\T)$, is given as $-\frac{1}{2}\big(\log \sigma_t^2(\T) +X^2_t/\sigma_t^2(\T)\big),$ where $\{\sigma_t^2(\T) | t \in \mathbb{Z} \}$ is defined as the solution of 
\begin{eqnarray*}
		 \sigma^{2}_{t}(\T):= \sigma_t^{2} = \omega+\sum_{i=1}^{p} \alpha_{i} X^2_{t-i} + \sum_{j=1}^{q} \beta_{j} \sigma^{2}_{t-j}.
\end{eqnarray*}
We consider the following standard assumptions.
\begin{enumerate}
    \renewcommand{\labelenumi}{\bf{G1.}}
	\item $\Theta$ is a compact set.
	\renewcommand{\labelenumi}{\bf{G2.}}
	\item For all $\theta \in \Theta$, $\sum_{j=1}^{q} \beta_{j} <1$.
	\renewcommand{\labelenumi}{\bf{G3.}}
	\item If $q>0$, $\mathcal{A}_{\theta_{0}} \mathscr(z)$ and $\mathcal{B}_{\theta_{0}} \mathscr(z)$ have no common root, $\mathcal{A}_{\theta_{0}} \mathscr(1) \neq 0$ and $\alpha_{0p}+\beta_{0q} \neq 0$ where $\mathcal{A}_{\T_0} \mathscr(z) = \sum_{i=1}^{p} \alpha_{i} \mathscr(z)^{i}$ and $\mathcal{B}_{\T_0} \mathscr(z) = 1-\sum_{j=1}^{q} \beta_{j} \mathscr(z)^{i}$. (Conventionally, $\mathcal{A}_{\T_0} \mathscr(z) = 0$ if $p=0$ and $\mathcal{B}_{\T_0} \mathscr(z)  =1$ if $q=0$).
	\renewcommand{\labelenumi}{\bf{G4.}}
	\item $\T_0$ is in the interior of $\Theta$.
\end{enumerate}
Under the assumptions above, \cite{francq:zakoian:2004} showed the strong consistency and the asymptotic normality of the estimator. Also, $\mathcal{J}= \E\left[ \partial^{2}_{\theta \theta'} l_t(\theta_{0}) \right]$ and $\mathcal{I}= \E\left[ \partial_{\theta}l_t(\theta_{0})\partial_{\theta'}l_t(\theta_{0}) \right]$  are  positive definite.  {\bf C2} and {\bf C3} can be shown by using Lemma \ref{lm_garch_C2}. Further, ${\bf C4}$ is implied by Lemma \ref{lm_garch_C4}.  Hence, we have the following result for the GARCH models.
\begin{theorem} \label{thm-garch}
	Suppose that the assumptions $\bf{G1}$-$\bf{G4}$ hold. If $\hat V_n$ is a consistent estimator of $V(\T_0)$, then under $H_{0}$, we have
\[ T_n' \hat V_n^{-1} T_n \stackrel{d}{\longrightarrow} \chi^2_q.\]
\end{theorem}

\subsection{First-order DAR model }\label{Sec:3_3}
Consider the following DAR(1) model: 
\begin{eqnarray} \label{dar-model}
			X_t=\phi X_{t-1}+e_t \sqrt{\omega+\alpha X_{t-1}^2},	
\end{eqnarray}
where $\omega$, $\alpha >0$ and $\{e_t\}$ is a sequence of i.i.d random variables with zero mean and unit variance. Denote the parameter vector by $\T=(\phi,\omega,\alpha)'$ and assume that the parameter space is given as follows:  
\begin{eqnarray}\label{dar_par}
   \Theta=\left\{ \T\in\mathbb{R}^3\,|\, \E\log|\phi+e_t\sqrt{\alpha}|<0 \ {\rm with} \ |\phi|\le c_1, c_2\le\omega\le c_3, \ {\rm and}\ c_4\le\alpha\le c_5 \right\}  
\end{eqnarray}
for some finite positive constants $c_1,c_2,c_3,c_4,$ and $c_5$. This constraint on the parameter space  can be found in \cite{ling:2004}. Under the null hypothesis that  $\epsilon_t$ follows $N(0,1)$, the equation (\ref{dar-model}) has a strictly stationary and ergodic solution. 
The QMLE for the DAR(1) model is given as follows:
\begin{eqnarray}\label{dar-qmle}
	\hat{\T}=\argmax_{\theta \in \Theta}\sum_{t=1}^{n}l_t(\T)
\end{eqnarray}
where 
\[ l_t(\T)=-\frac{1}{2}\log\left( \omega+\alpha X_{t-1}^2 \right) - \frac{\left(X_t-\phi X_{t-1}\right)^2}{2\left(\omega+\alpha X_{t-1}^2\right)}. \]
\cite{ling:2004} established the strong consistency and the asymptotic normality of the QMLE above. Hence, in order to implement the IM test for the DAR(1) model,  we only need to verify whether  ${\bf C1}$-${\bf C3}$ are satisfied. It is evident that ${\bf C1}$ holds, and Lemma \ref{lm_dar_C2} confirms the validity of ${\bf C2}$ and ${\bf C3}$. One can therefore see that Theorem \ref{thm_test} holds for the DAR(1) model above.


\begin{table}[t!]
	\renewcommand\arraystretch{1.25}
	\tabcolsep=10pt
	\caption{Empirical sizes of the IM$_{opt}$ test and other normality tests for the TMA(1) model}
	\label{tma-size}
	\centering
	\begin{tabular}{cccccccc}
		\hline
	    &          & \multicolumn{ 2}{c}{$\T=(0.2,0.7,1)$} & \multicolumn{ 2}{c}{$\T=(0.9,-0.7,1)$} & \multicolumn{ 2}{c}{$\T=(-0.5,1,1)$} \\
		\hline
	Test	&     $n$        &  $\alpha$=5\% & $\alpha$=10\% &  $\alpha$=5\% & $\alpha$=10\% &  $\alpha$=5\% & $\alpha$=10\% \\
		\hline
		&       1000 &     0.045  &     0.096  &     0.044  &     0.093  &     0.062  &     0.099  \\
		
		IM$_{opt}$ &       2000 &     0.044  &     0.091  &     0.052  &     0.098  &     0.041  &     0.089  \\
		
		&       3000 &     0.058  &     0.112  &     0.056  &     0.105  &     0.058  &     0.105  \\
		\hline
		&       1000 &     0.052  &     0.095  &     0.048  &     0.092  &     0.044  &     0.091  \\
		
		JB &       2000 &     0.043  &     0.089  &     0.058  &     0.104  &     0.048  &     0.089  \\
		
		&       3000 &     0.057  &     0.107  &     0.045  &     0.103  &     0.054  &     0.101  \\
		\hline
		&       1000 &     0.034  &     0.073  &     0.030  &     0.063  &     0.030  &     0.065  \\
		
		KS &       2000 &     0.048  &     0.080  &     0.035  &     0.070  &     0.034  &     0.073  \\
		
		&       3000 &     0.032  &     0.069  &     0.030  &     0.054  &     0.028  &     0.066  \\
		\hline
		&       1000 &     0.049  &     0.106  &     0.047  &     0.095  &     0.056  &     0.104  \\
		
		CVM &       2000 &     0.048  &     0.095  &     0.046  &     0.101  &     0.058  &     0.102  \\
		
		&       3000 &     0.053  &     0.113  &     0.049  &     0.088  &     0.051  &     0.100  \\
		\hline
		&       1000 &     0.045  &     0.108  &     0.046  &     0.093  &     0.055  &     0.106  \\
		
		AD &       2000 &     0.047  &     0.093  &     0.048  &     0.105  &     0.058  &     0.100  \\
		
		&       3000 &     0.054  &     0.110  &     0.050  &     0.090  &     0.052  &     0.101  \\
		\hline
		&       1000 &     0.040  &     0.096  &     0.044  &     0.107  &     0.054  &     0.106  \\
		
		LL &       2000 &     0.047  &     0.099  &     0.046  &     0.107  &     0.050  &     0.111  \\
		
		&       3000 &     0.045  &     0.108  &     0.046  &     0.101  &     0.043  &     0.102  \\
		\hline
	\end{tabular}  
\end{table}

\section{Simulation studies} \label{Sec:simu}
\indent
    We shall evaluate performance of the proposed test for the TMA(1) model, the GARCH(1,1) model and the DAR(1) model, respectively. For comparisons, we also conduct the following normality tests based on residuals: the Jarque-Bera (JB) test, the Kolmogorov–Smirnov (KS) test, the Cram{\'e}r–von Mises (CVM) test, the Anderson–Darling (AD) test, and the Lilliefors (LL) test. To the best of the author's knowledge, the limiting null distributions of these tests based on residuals have not been established, except for the JB test for GARCH models (cf. \cite{kulperger2005high}). Nevertheless, we use these tests assuming that residuals behave like i.i.d. random variables. The Shapiro–Wilk test and the D’Agostino-Pearson test were also considered, but their results did not show significant differences compared to the AD test and the JB test, respectively. So, we do not report them.
	
	Under $H_0$, we generate errors from $N(0,1)$. To evaluate empirical powers, we consider the following error distributions under $H_1$: the $t$-distribution with 15 degrees of freedom (t(15)), the centered logistic distribution (LD), the normal mixture distributions of $0.2N(0,2)+0.8N(0,0.75)$ (NM1), $0.5N(0.7,1)+0.5N(-0.7,1)$ (NM2), and  $0.5N(1,2)+0.5N(-1,2)$ (NM3), and the generalized lambda distribution (GLD) with the parameter of $(\lambda_1,\lambda_2,\lambda_3, \lambda_4)=(0,1,0.2,0.2)$. All the distributions considered under $H_1$ are scaled to have unit variance. Here,  it is important to note that each kurtosis of t(15), LD, and NM1 is greater than 3, indicating that these distributions have fatter tails than the normal distribution, whereas NM2, NM3, and GLD have a kurtosis less than 3.  

 The following empirical sizes and powers are calculated based on 2,000 repetitions. We performed the IM test with every combination of the elements in (\ref{elm}), but we report the results of the IM test with optimal subset that produced the best performance. For each model considered below, empirical sizes are presented in tables, and empirical powers obtained at the significance level of 10\% are displayed in figures.


We first consider the TMA(1) model with the parameter of $\T=(\phi,\xi,\sigma^2)$ as follows:
	\begin{eqnarray*}
	    X_t = (\phi+\xi I(X_{t-1}\le 0.5))\sigma e_{t-1}+\sigma e_t,
	\end{eqnarray*}
where we consider $\T=(0.2,0.7, 1)$, $(0.9,-0.7, 1)$ and $(-0.5,1, 1)$. 
\begin{figure}[htp] 
  \centering 
  \includegraphics[width=\textwidth,height=0.7\textheight]{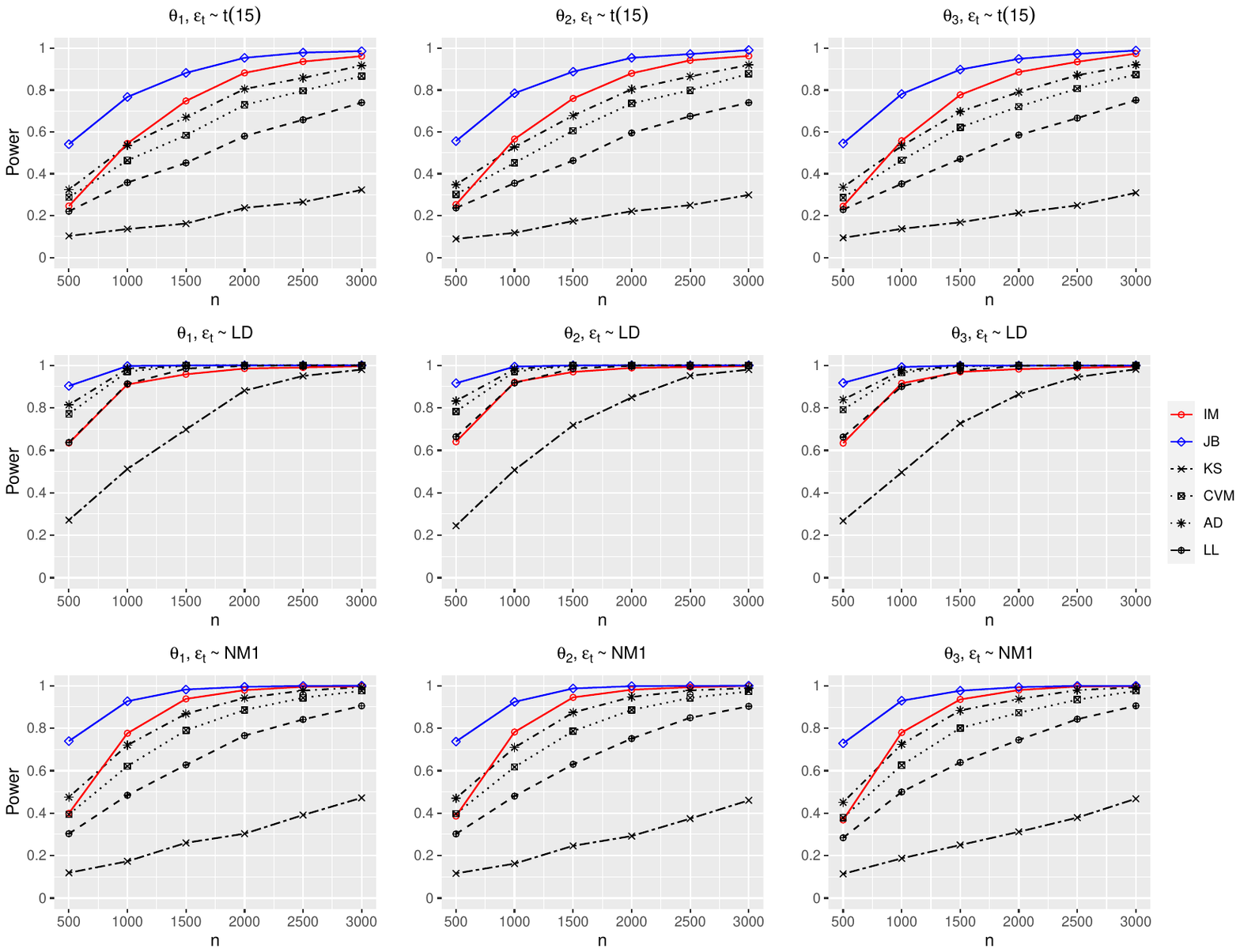} 
  \caption{Empirical powers of the IM$_{opt}$ test and other normality tests for the TMA(1) model when $\ep_t$ follows t(15), LD, and NM1, respectively.} 
  \label{fig:tma-power1} 
\end{figure}

\begin{figure}[htp] 
  \centering 
  \includegraphics[width=\textwidth,height=0.7\textheight]{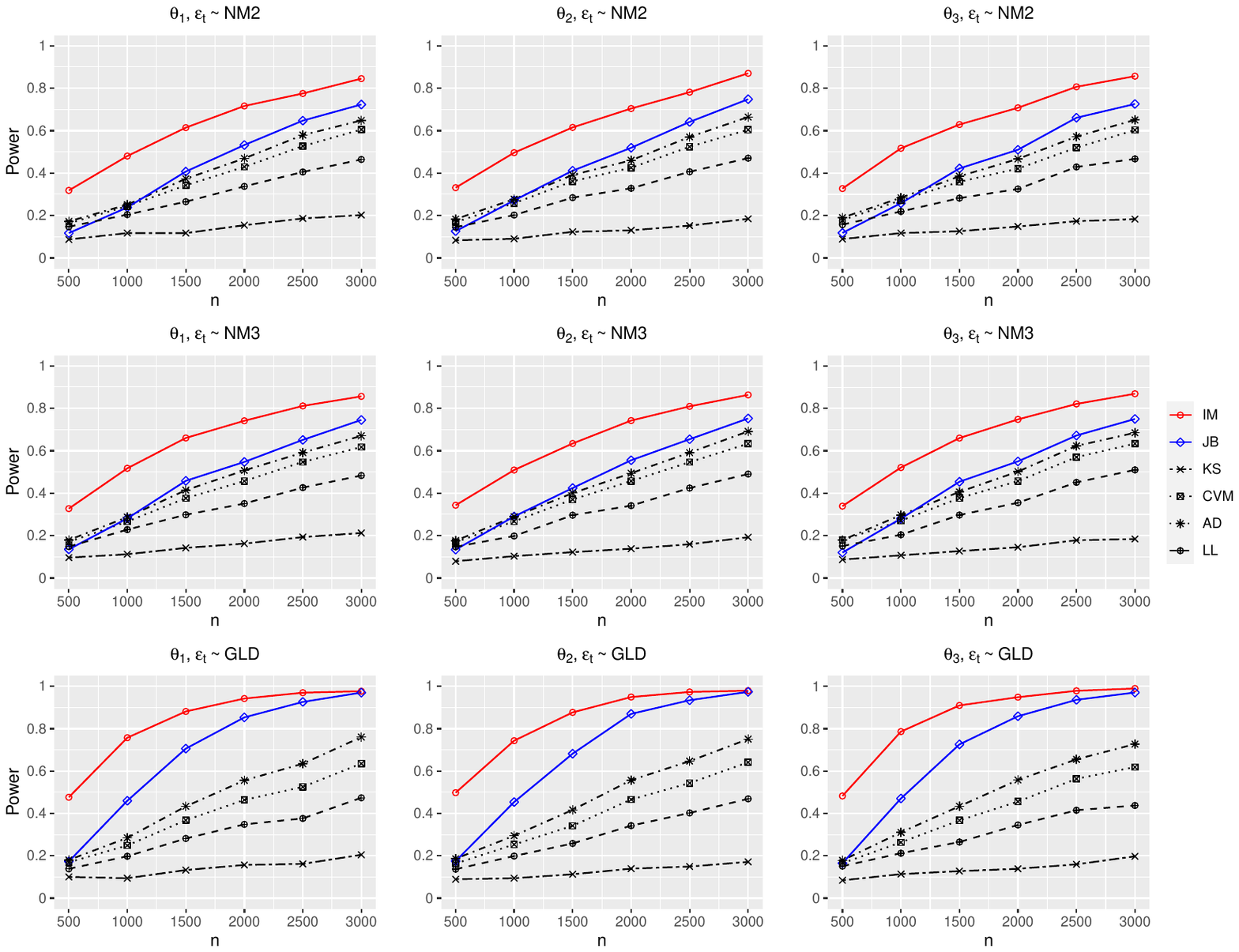} 
\caption{Empirical powers of the IM$_{opt}$ test and other normality tests for the TMA(1) model when $\ep_t$ follows NM2, NM3, and GLD, respectively.}  
\label{fig:tma-power2} 
\end{figure}

    For the above TMA(1) model, the IM test based on $\tilde d(X_t;\T) =\pa^2_{\T_3 \T_3} \tilde l ( X_{t};\theta ) + \pa_{\T_3 } \tilde l ( X_{t};\theta ) \pa_{\T_3 } \tilde l ( X_{t};\theta )$, say IM$_{opt}$, exhibits the best performance, where $\T_3=\sigma^2$. The empirical sizes and powers for the  IM$_{opt}$ and other normality tests are presented in Table \ref{tma-size} and Figures \ref{fig:tma-power1}-\ref{fig:tma-power2}, respectively. It can be seen in Table \ref{tma-size} that the  IM$_{opt}$ consistently produces proper sizes for all parameter cases.  As mentioned earlier, the limiting distributions of the other residual-based tests have not been established for the TMA(1) models. It is, however, noteworthy that these tests still yield reasonable sizes, except for the KS test, which produces somewhat smaller sizes compared to the specified significance levels. From Figure \ref{fig:tma-power1}, we can also see that the IM$_{opt}$ test produces typical shapes of power curves, with the empirical powers increasing as the sample size $n$ grows. Similar trends are observed in other tests but the KS test and the LL test yield comparatively lower powers. Although the JB test performs best for the cases of t(15), LD, and NM1 distributions, the IM$_{opt}$ test also perform quite well in these cases. It is widely recognized  in the literature that the JB test shows strong performance for the cases of heavy-tailed distributions (cf.\cite{thadewald:2007}). Our IM$_{opt}$ test, however, outperforms other normality tests for NM2, NM3, and GLD, as can be seen in Figures \ref{fig:tma-power2}.

\begin{table}[t!]
	\renewcommand\arraystretch{1.25}
	\tabcolsep=10pt
	\caption{Empirical sizes of the IM$_{opt}$ test and other normality tests for the GARCH(1,1) model}
	\label{garch-size}
	\centering
	\begin{tabular}{cccccccc}
		\hline
	    &           & \multicolumn{ 2}{c}{$\T=(0.2,0.3,0.2)$} & \multicolumn{ 2}{c}{$\T=(0.2,0.1,0.8)$} & \multicolumn{ 2}{c}{$\T=(0.2,0.05,0.9)$} \\
		\hline
	Test	&     $n$       &  $\alpha$=5\% & $\alpha$=10\% &  $\alpha$=5\% & $\alpha$=10\% &  $\alpha$=5\% & $\alpha$=10\% \\
		\hline
		&       1000 &     0.038  &     0.088  &     0.044  &     0.915  &     0.040  &     0.092  \\
		
		IM$_{opt}$ &       2000 &     0.057  &     0.105  &     0.049  &     0.096  &     0.054  &     0.104  \\
		
		&       3000 &     0.050  &     0.091  &     0.055  &     0.109  &     0.049  &     0.097  \\
		\hline
		&       1000 &     0.040  &     0.088  &     0.051  &     0.090  &     0.048  &     0.090  \\
		
		JB &       2000 &     0.043  &     0.087  &     0.051  &     0.097  &     0.046  &     0.099  \\
		
		&       3000 &     0.041  &     0.089  &     0.049  &     0.084  &     0.053  &     0.098  \\
		\hline
		&       1000 &     0.036  &     0.087  &     0.036  &     0.074  &     0.046  &     0.094  \\
		
		KS &       2000 &     0.045  &     0.085  &     0.050  &     0.091  &     0.040  &     0.084  \\
		
		&       3000 &     0.047  &     0.091  &     0.044  &     0.094  &     0.036  &     0.074  \\
		\hline
		&       1000 &     0.048  &     0.102  &     0.045  &     0.092  &     0.043  &     0.093  \\
		
		CVM &       2000 &     0.051  &     0.105  &     0.049  &     0.093  &     0.047  &     0.091  \\
		
		&       3000 &     0.053  &     0.102  &     0.051  &     0.095  &     0.041  &     0.097  \\
		\hline
		&       1000 &     0.048  &     0.102  &     0.046  &     0.092  &     0.048  &     0.095  \\
		
		AD &       2000 &     0.049  &     0.106  &     0.045  &     0.100  &     0.046  &     0.093  \\
		
		&       3000 &     0.057  &     0.107  &     0.051  &     0.097  &     0.042  &     0.096  \\
		\hline
		&       1000 &     0.051  &     0.114  &     0.040  &     0.099  &     0.040  &     0.088  \\
		
		LL &       2000 &     0.059  &     0.113  &     0.045  &     0.102  &     0.049  &     0.103  \\
		
		&       3000 &     0.047  &     0.126  &     0.045  &     0.099  &     0.041  &     0.105  \\
		\hline
	\end{tabular} 
\end{table}

We also examine the performance of the IM test for the GARCH(1,1) model and the DAR(1) model, given in (\ref{GARCH}) with $p=q=1$ and (\ref{dar-model}), respectively. The optimal IM test obtained for the GARCH(1,1) model is based on $\tilde{d}(X_t;\T)=(\tilde{d}_{11}(X_t;\T),\tilde{d}_{22}(X_t;\T))^{'}$, where $\tilde d_{ij}(X_t;\T) =\pa^2_{\T_i \T_j} \tilde l ( X_{t};\theta ) + \pa_{\T_i} \tilde l ( X_{t};\theta) \pa_{\T_j} \tilde l ( X_{t};\theta)$ and $\T=(\T_1,\T_2,\T_3)=(\omega, \alpha,\beta)$. Meanwhile, for the DAR(1) model, the IM test using $\tilde{d}(X_t;\T)=\pa^2_{\T_2 \T_3} \tilde l ( X_{t};\theta ) + \pa_{\T_2} \tilde l ( X_{t};\theta) \pa_{\T_3} \tilde l ( X_{t};\theta)$ shows the best performance, where $\T=(\T_1,\T_2,\T_3)=(\phi,\omega,\A)$.

The parameters considered are $\T=(0.2, 0.3, 0.2), ( 0.2, 0.1, 0.8)$, and $( 0.2, 0.05, 0.9)$ for the GARCH model and $\T=(0.2,0.5,0.3), (0.4,0.5,0.5)$, and $(0.5,0.5,0.7)$ for the DAR model. 
The empirical sizes for the GARCH(1,1) model and the DAR(1) model are presented in Tables \ref{garch-size} and \ref{dar-size}, respectively. One can see that the IM$_{opt}$ test consistently achieves reasonable sizes. Our test exhibits stable sizes even in the highly persistent scenarios, such as when $\alpha+\beta = 0.95$ for the GARCH model and when $\phi=0.5$ and $\A=0.7$ for the DAR(1) model. However, the JB test produces relatively larger sizes in the last parameter case for the DAR(1) model. Most of the other tests yield proper sizes. The empirical powers are displayed in Figures \ref{fig:garch-power1}-\ref{fig:garch-power2} for the GARCH(1,1) model and in Figures \ref{fig:dar-power1}-\ref{fig:dar-power2} for the DAR(1) model. The results obtained are similar to those for the TMA(1) model discussed earlier.  

Overall, our simulation results strongly support the validity and effectiveness of the IM test, particularly in cases where the error distributions are not heavy-tailed, while also showing good performance in other heavy-tailed cases. Based on these findings, we can conclude that the IM test serves as a valuable complement to existing tests for testing the normality of innovations in time series models.

\begin{figure}[htp] 
  \centering 
  \includegraphics[width=\textwidth,height=0.7\textheight]{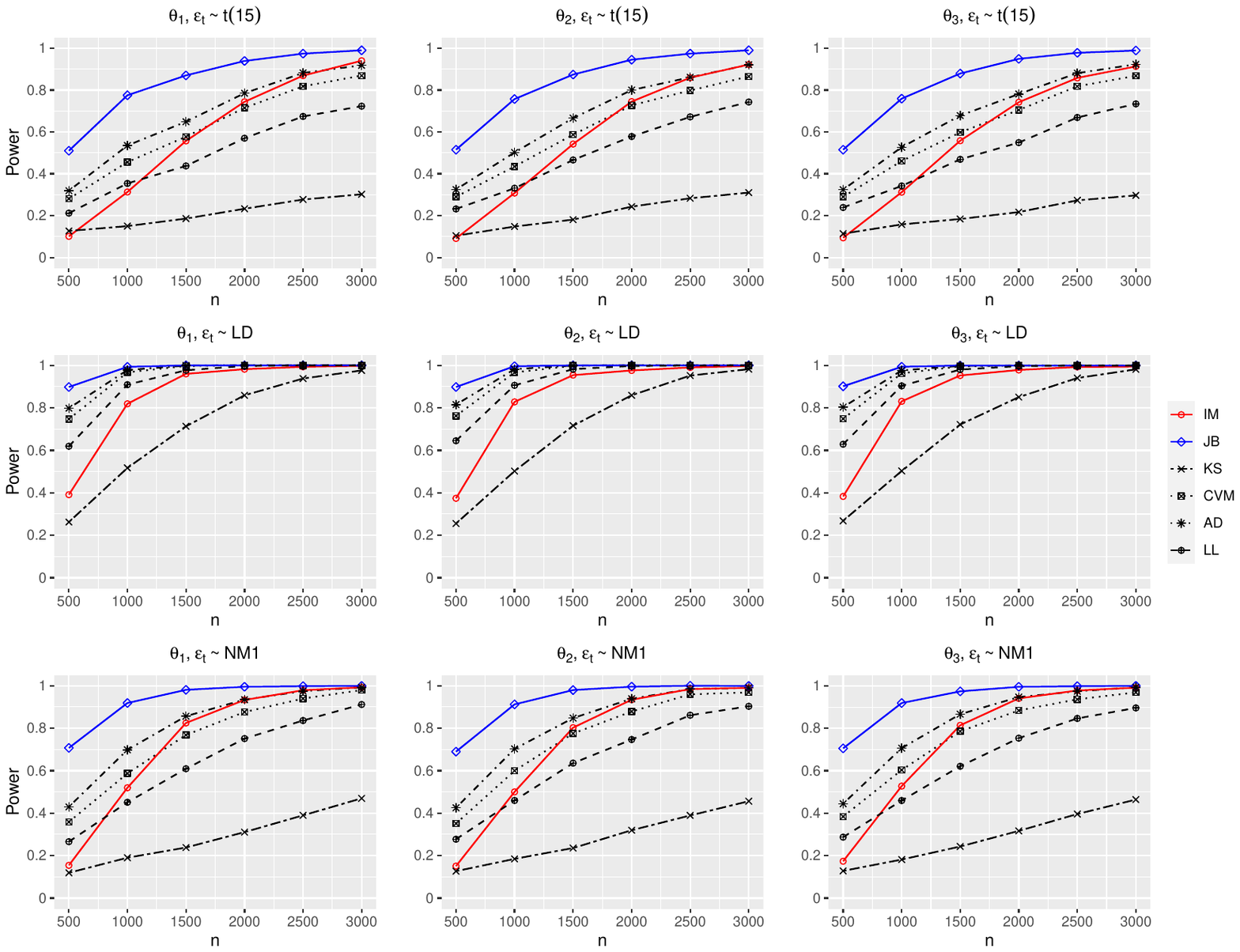} 
\caption{Empirical powers of the IM$_{opt}$ test and other normality tests for the GARCH(1,1) model when $\ep_t$ follows t(15), LD, and NM1, respectively.}   \label{fig:garch-power1} 
\end{figure}

\begin{figure}[htp] 
  \centering 
  \includegraphics[width=\textwidth,height=0.7\textheight]{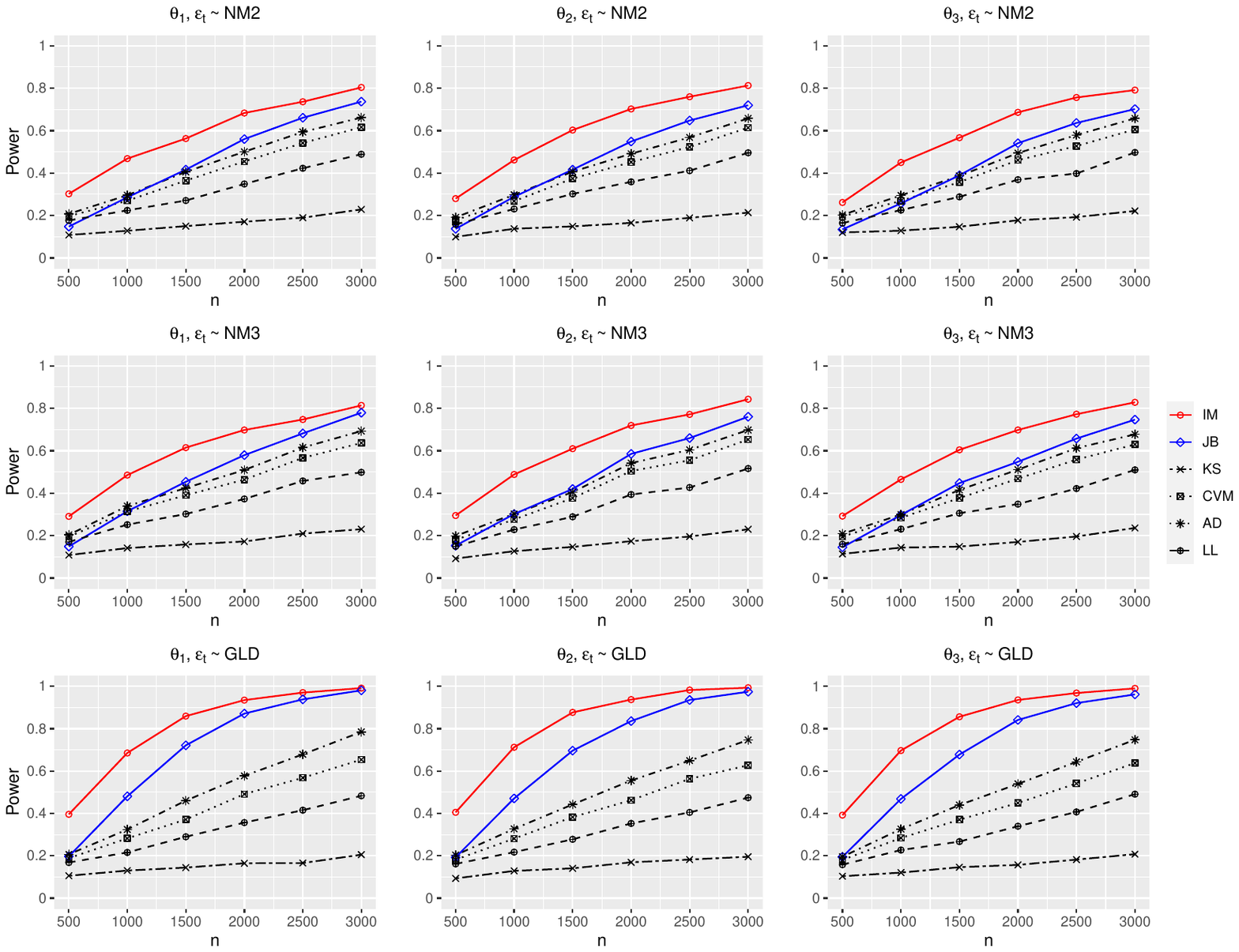} 
\caption{Empirical powers of the IM$_{opt}$ test and other normality tests for the GARCH(1,1) model when $\ep_t$ follows NM2, NM3, and GLD, respectively.}  
  \label{fig:garch-power2} 
\end{figure}

\begin{table}[t!]
		\renewcommand\arraystretch{1.25}
		\tabcolsep=10pt
\caption{Empirical sizes of the IM$_{opt}$ test and other normality tests for the DAR(1) model}		\label{dar-size}
		\centering
	    \begin{tabular}{cccccccc}
	    	\hline
	    	Size &          n & \multicolumn{ 2}{c}{(0.2,0.5,0.3)} & \multicolumn{ 2}{c}{(0.4,0.5,0.5)} & \multicolumn{ 2}{c}{(0.5,0.5,0.7)} \\
	    	\hline
	    	&            &  $\alpha$=5\% & $\alpha$=10\% &  $\alpha$=5\% & $\alpha$=10\% &  $\alpha$=5\% & $\alpha$=10\% \\
	    	\hline
	    	           &       1000 &     0.051  &     0.103  &     0.057  &     0.103  &     0.057  &     0.099  \\

           IM$_{opt}$ &       2000 &     0.046  &     0.102  &     0.059  &     0.107  &     0.060  &     0.119  \\

           &       3000 &     0.057  &     0.104  &     0.052  &     0.104  &     0.063  &     0.108  \\
	    	\hline
	    	&       1000 &     0.041  &     0.090  &     0.059  &     0.107  &     0.099  &     0.151  \\
	    	
	    	JB &       2000 &     0.061  &     0.103  &     0.056  &     0.104  &     0.090  &     0.141  \\
	    	
	    	&       3000 &     0.048  &     0.087  &     0.059  &     0.099  &     0.090  &     0.144  \\
	    	\hline
	    	&       1000 &     0.042  &     0.086  &     0.046  &     0.084  &     0.038  &     0.074  \\
	    	
	    	KS &       2000 &     0.039  &     0.082  &     0.040  &     0.093  &     0.048  &     0.090  \\
	    	
	    	&       3000 &     0.032  &     0.079  &     0.046  &     0.084  &     0.048  &     0.088  \\
	    	\hline
	    	&       1000 &     0.048  &     0.094  &     0.047  &     0.098  &     0.067  &     0.110  \\
	    	
	    	CVM &       2000 &     0.050  &     0.101  &     0.047  &     0.101  &     0.057  &     0.106  \\
	    	
	    	&       3000 &     0.040  &     0.087  &     0.054  &     0.096  &     0.053  &     0.093  \\
	    	\hline
	    	&       1000 &     0.047  &     0.092  &     0.043  &     0.102  &     0.067  &     0.114  \\
	    	
	    	AD &       2000 &     0.053  &     0.101  &     0.050  &     0.098  &     0.069  &     0.111  \\
	    	
	    	&       3000 &     0.035  &     0.087  &     0.053  &     0.093  &     0.057  &     0.107  \\
	    	\hline
	    	&       1000 &     0.043  &     0.102  &     0.048  &     0.099  &     0.057  &     0.119  \\
	    	
	    	LL &       2000 &     0.048  &     0.108  &     0.048  &     0.105  &     0.048  &     0.104  \\
	    	
	    	&       3000 &     0.036  &     0.104  &     0.038  &     0.095  &     0.042  &     0.102  \\
	    	\hline
	    \end{tabular} 
	\end{table}

\begin{figure}[htp] 
  \centering 
  \includegraphics[width=\textwidth,height=0.7\textheight]{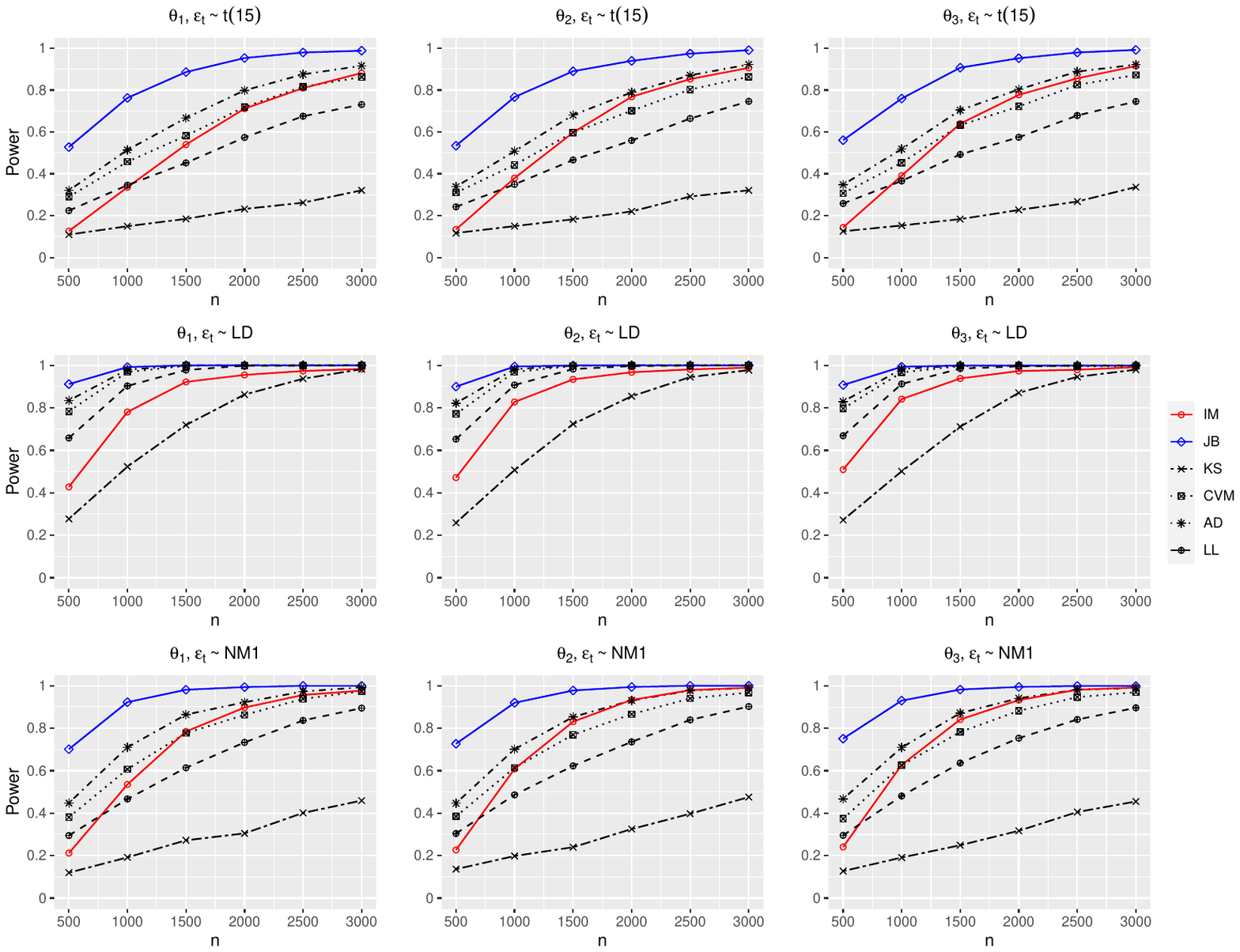} 
\caption{Empirical powers of the IM$_{opt}$ test and other normality tests for the DAR(1) model when $\ep_t$ follows t(15), LD, and NM1, respectively.}  \label{fig:dar-power1} 
\end{figure}

\begin{figure}[htp] 
  \centering 
  \includegraphics[width=\textwidth,height=0.7\textheight]{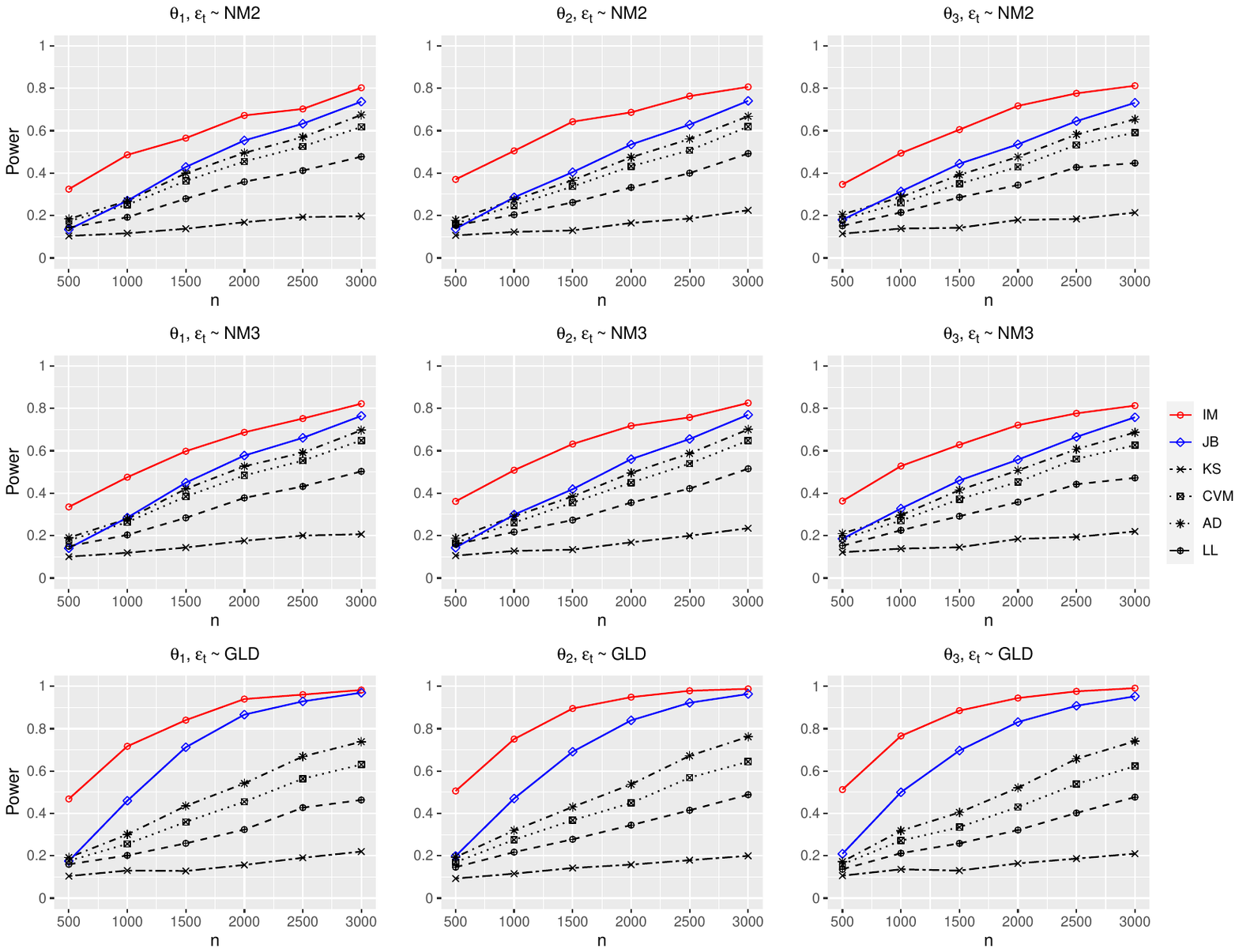} 
\caption{Empirical powers of the IM$_{opt}$ test and other normality tests for the DAR(1) model when $\ep_t$ follows NM2, NM3, and GLD, respectively.}  
  \label{fig:dar-power2} 
\end{figure}

\section{Real data analysis}\label{Sec:real}

In this section, we analyze the log return series of the S\&P500 index during two distinct periods: 2001-2005 and 2006-2010, consisting of 1255 and 1258 observations, respectively. Figure \ref{fig:orig-series} presents the original index series (L) for each period and their corresponding log return series (R).   During the first period (2001-2005), both the market and the economy exhibited relative stability.  In such circumstances, models with normal innovations are typically sufficient to fit the data. However, it is needed to note that the market experienced the global financial crisis during the 2006-2010 period. In this case, it is well known that  heavy-tailed distributions are more appropriate as error distribution.
   
\begin{figure}[t]
	\centering
    \subfigure{
    	\includegraphics[width=3.1in,height=1.8in]{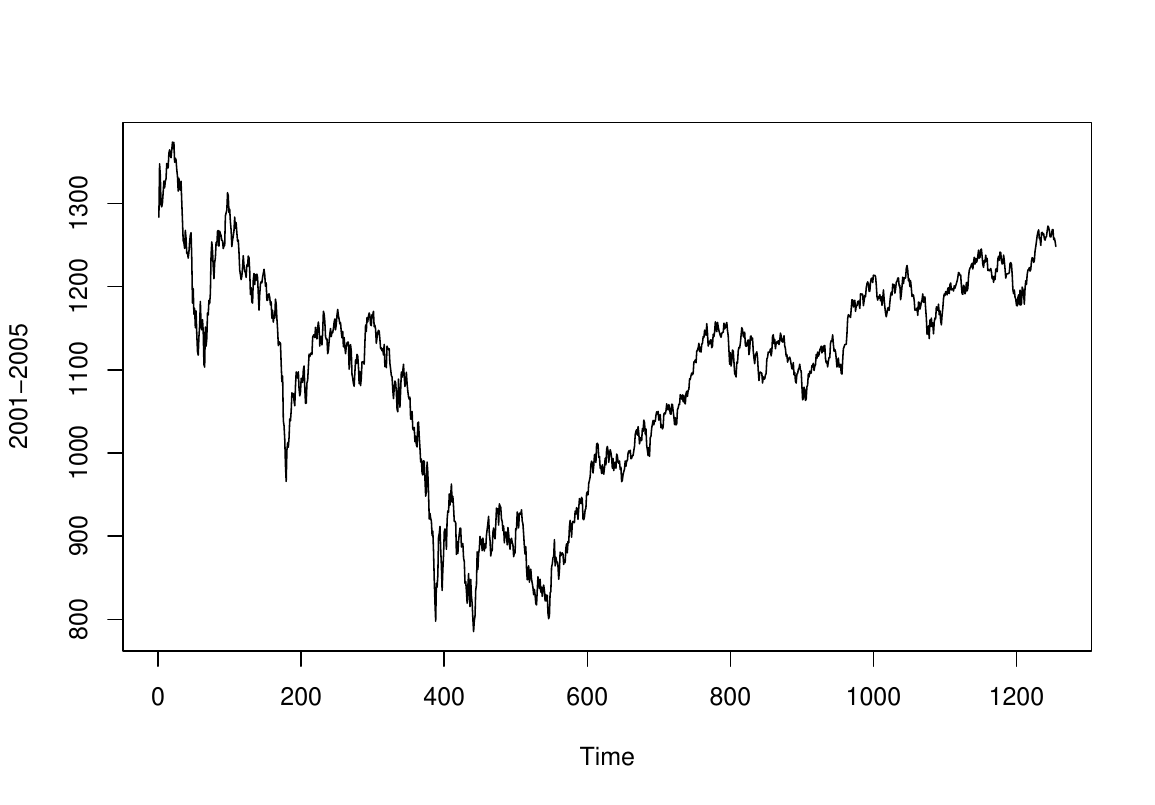}
    }\vspace{-4mm}
    \subfigure{
    	\includegraphics[width=3.1in,height=1.8in]{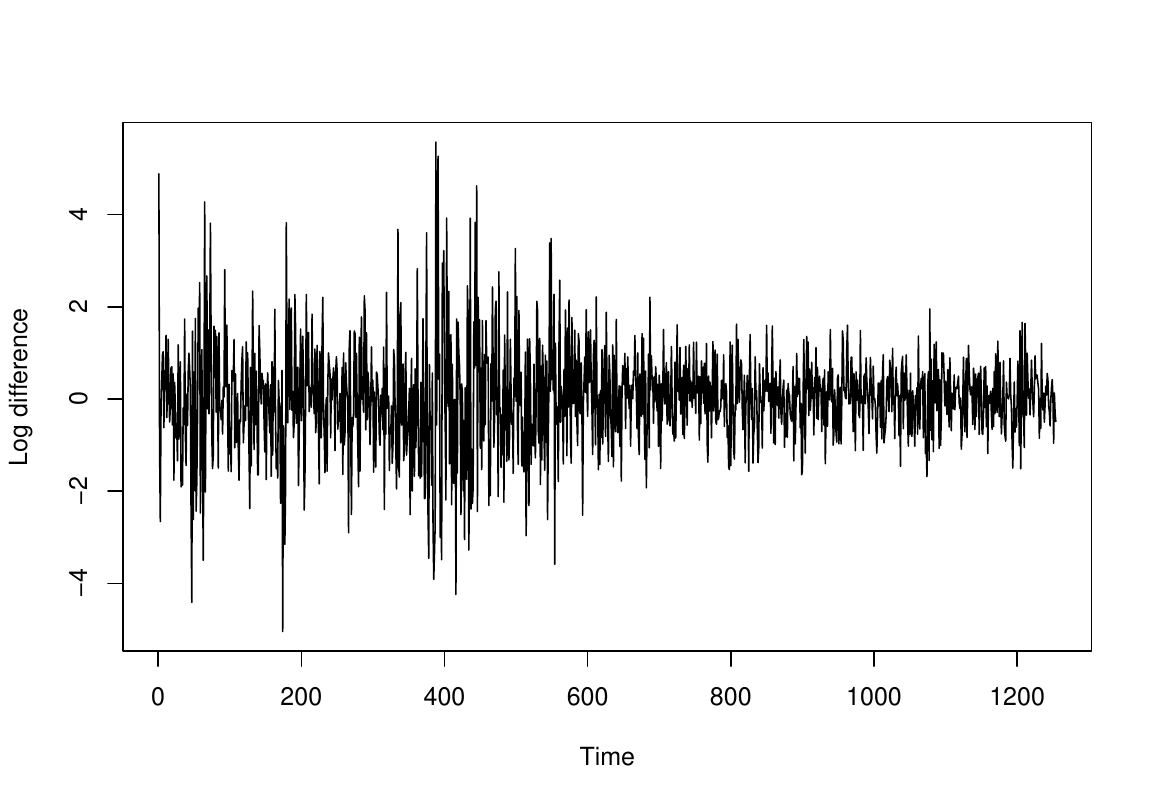}
    }
    \subfigure{
    	\includegraphics[width=3.1in,height=1.8in]{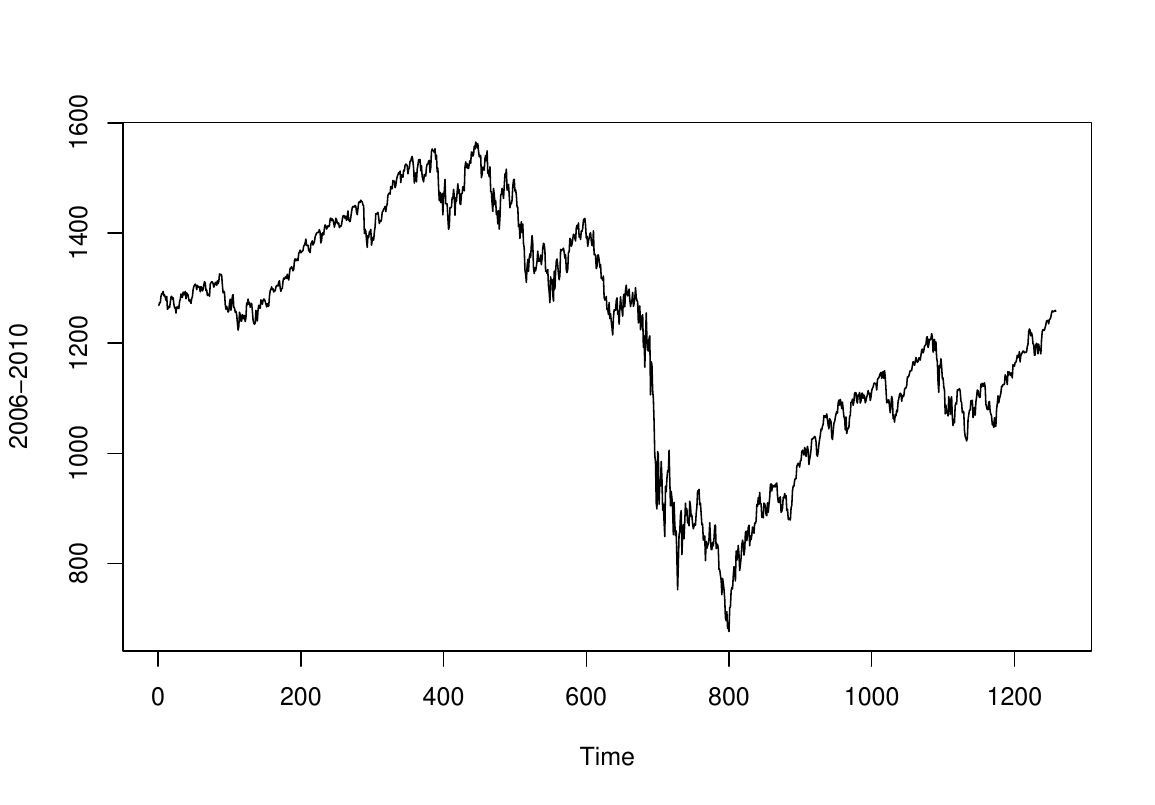}
    }
    \subfigure{
    	\includegraphics[width=3.1in,height=1.8in]{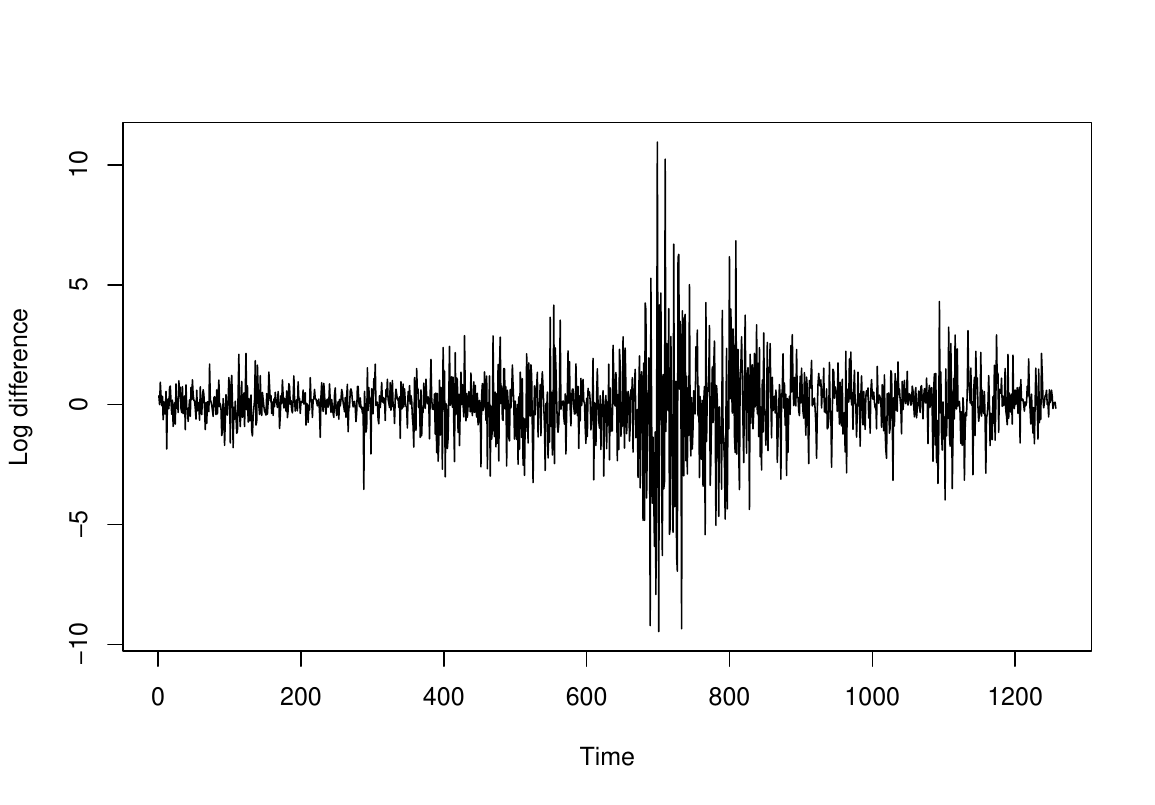}
    }
	\caption{The S\&P500 index series (L) and the corresponding log returns series (R)}
	\label{fig:orig-series}
\end{figure}

 As commonly used in empirical practice, we fit the GARCH(1,1) model given in (\ref{GARCH}) to each data set since each series shows typical features such as arch effect and also due to its simplicity.  Based on the simulation results above, we perform the IM test with $\tilde{d}(X_t;\T)=(\tilde{d}_{11}(X_t;\T),\tilde{d}_{22}(X_t;\T))^{'}$, and also conduct the residual-based JB test. The estimation results and the p-values of the IM test and the JB test are presented in Table \ref{fitted-garch}. The $\hat{\omega}$, $\hat{\alpha}$ and $\hat{\beta}$ are QML estimates and the asymptotic standard errors are given in parentheses. We first note that both tests yield the same conclusion for the second period. The p-values obtained from both tests are close to zero, indicating the rejection of the null hypothesis. As expected, the normal distribution is not suitable as the error distribution during this period. However,for the first period, the two tests lead to different conclusions. The IM test produces a p-value of 0.241, while the JB test yields a p-value of 0.025. That is, at the significance level of 5\%, the IM test does not reject the null hypothesis whereas the JB test rejects it. The JB test relies on skewness and kurotosis, which are sensitive to outlying observations. Upon inspecting the residuals, we omit the residual with the minimum value of -4.33, which is considered to be an influential point, and  reconduct the JB test for the remaining residuals. The resulting p-value of the JB test increases significantly from 0.025 to 0.49, suggesting that the previous result of the JB test is influenced by the presence of the minimum residual. For comparison, we also reimplement the IM test for the log return data without including the observation corresponding to the omitted residual. The p-value obtained from the IM test is 0.256. Based on these results, we can presume that the normal innovation is suitable for the first period.

\begin{table}[t]
	\caption{Parameter estimates and the results of the IM$_{opt}$ test and the JB test}
	\label{fitted-garch}
 \centering
	\small
	\setlength{\tabcolsep}{2mm}{
 \renewcommand\arraystretch{1.8}
	\begin{tabular}{ccccccc}
		\hline
		Periods & $n$ & $\hat{\omega}$ & $\hat{\alpha}$ & $\hat{\beta}$ & IM$_{opt}$ & JB \\
		\hline
		2001\,--\,2005 & 1255 & 0.006(0.004) & 0.066(0.015) & 0.928(0.015) & 0.241 & 0.025 \\
		2006\,--\,2010 & 1258 & 0.017(0.005) & 0.092(0.012) & 0.899(0.012)& 0.001 &0.000\\
		\hline
	\end{tabular}
    }
\end{table}

\section{Concluding remarks}
In this study, we introduced the IM test for testing the normality of innovations in time series models, and provided a set of conditions for time series models under which the IM test follows a chi-square distribution as its limiting null distribution. We applied the IM test to the TMA(1) model, the GARCH model, and DAR(1) model as examples. Through simulation study and real data analysis, we demonstrated the validity and usefulness of the test. It is important to  note that there is no objective criterion for selecting an optimal set of elements to be used in the IM test. Therefore, in order to use the test in  practice, a preliminary simulation would be required to  determine the optimal IM test.  

We expect that the IM test procedure can be extended to multivariate time series models. Application to a random coefficient model is also of interest, as  residuals are not easy to be obtained due the presence of random terms in a random coefficient, consequently making it cumbersome to use the residual-based normality tests.  We leave these issues as a possible topic of future study.


\section{Appendix}
In this appendix, we shall provide the proofs of Theorems and some Lemmas. 

\subsection{Proofs for Section \ref{Sec:2}}
\noindent{\bf Proof of Theorem \ref{thm_qmle}}\\
By {\bf A1} and {\bf A3}(a), one can see that
\[ \sup_{\theta\in\Theta} \Big| \frac{1}{n}\sum_{t=1}^n l(X_t;\T) -\E l(X_1;\T) \Big| =o(1)\quad a.s.\]
(cf. \cite{straumann:2006}). Using the above and {\bf A3}(b), we have
\[ \sup_{\theta\in\Theta} \Big| \frac{1}{n}\sum_{t=1}^n \tilde l(X_t;\T) -\E l(X_1;\T) \Big| =o(1)\quad a.s.,\]
and thus, by the standard arguments, the strong consistency of $\hat \T_n$ is established from {\bf A2}.

Next, we show the asymptotic normality. Since
\begin{eqnarray}\label{pl}
   \pa_\T l(X_t;\T)=-\frac{1}{2} \left(  \frac{1}{\sigma_t^2(\T) }\pa_\T \sigma_t^2(\T)
   - \frac{2}{\sigma_t^2(\T)}(X_t-\mu_t(\T))\pa_\T \mu_t (\T)
   -\frac{1}{\sigma_t^4(\T)} (X_t -\mu_t(\T))^2 \pa_\T \sigma_t^2(\T)\right),
\end{eqnarray}
we have
\begin{eqnarray}\label{pL}
   \pa_\T l(X_t;\T_0)
   &=&-\frac{1}{2} \left(  \frac{1}{\sigma_t^2(\T_0) }(1-\eta_t^2)\pa_\T \sigma_t^2(\T_0)
   - \frac{2}{\sigma_t(\T_0)} \eta_t\pa_\T \mu_t (\T_0)
   \right)
\end{eqnarray}
and thus we can see that  $\E\big[\pa_\T l(X_t;\T_0)|\mathcal{F}_{t-1}\big] =0$. Recalling that $\{\pa_\T l(X_t;\T_0)\}$ is strictly stationary and ergodic, it follows from the central limit theorem for martingales and {\bf A7} that
\begin{eqnarray}\label{clt}
   \frac{1}{\sqrt{n}} \sum_{t=1}^n \pa_\T \tilde l(X_t ;\T_0) \stackrel{d}{\longrightarrow} N_p({\bf 0},\mathcal{I}). 
\end{eqnarray}
Using Taylor's theorem, we have
\[0=\frac{1}{\sqrt{n}}\sum_{t=1}^n \pa_\T \tilde l(X_t;\hat\T_n)
=\frac{1}{\sqrt{n}}\sum_{t=1}^n \pa_\T \tilde l(X_t;\T_0)+\frac{1}{n}\sum_{t=1}^n \paa \tilde l(X_t;\T^*_n )\sqrt{n}(\hat\T_n-\T_0),\]
where $\T^*_n$ lies between $\hat\T_n$ and $\T_0$, and thus we can write that 
\[\sqrt{n}(\hat\T_n -\T_0)=-\mathcal{J}^{-1}\frac{1}{\sqrt{n}}\sum_{t=1}^n \pa_\T \tilde l(X_t;\T_0)-\mathcal{J}^{-1}(\tilde B_n -\mathcal{J})\sqrt{n}(\hat\T_n -\T_0), \]
where $\tilde B_n=\frac{1}{n}\sum_{t=1}^n \paa \tilde l(X_t;\T^*_n )$.
Further, using {\bf A6}, the continuity of $\paa l(X_t;\T)$ in $\T$, and the strong consistency of $\hat\T_n$, one can show that $\frac{1}{n}\sum_{t=1}^n \paa  l(X_t;\T^*_n )$ converges to $\mathcal{J}$ almost surely, so does $\tilde B_n$ due to {\bf A8}. $\sqrt{n}(\hat\T_n -\T_0)$ should therefore be $O_P(1)$, so we have
\begin{eqnarray}\label{P1}
\sqrt{n} (\hat{\theta}_{n}-\theta_0)
= -\mathcal{J}^{-1}\frac{1}{\sqrt{n}} \sum_{t=1}^n \pa_\T \tilde l(X_t ;\T_0)+o_P(1),
\end{eqnarray}
which together with (\ref{clt}) establishes the asymptotic normality of $\hat\T_n$.
\hfill{$\Box$}\\

\begin{lm}\label{lm1}
Under $H_0$, it holds that
\[ \E\left[ \partial_{\theta}l(X_{t};\theta_{0})\partial_{\theta'}l(X_{t};\theta_{0})\right]= -\E\left[ \partial^{2}_{\theta \theta'} l(X_{t};\theta_{0})\right].\]
\end{lm}
\begin{proof}
From (\ref{pL}), we have
\begin{eqnarray*}
   \pa_\T l(X_t;\T_0)\pa_{\T'} l(X_t;\T_0) = \frac{1}{4}  \bigg[ \frac{1}{\sigma_t^4(\T_0)} (1-\eta_t^2)^2 \pa_\T \sigma_t^2(\T_0) \pa_{\T'} \sigma_t^2(\T_0) + \frac{4}{\sigma_t^2(\T_0)} \eta_t^2 \pa_\T \mu_t(\T_0) \pa_{\T'} \mu_t(\T_0) \\
   - \frac{2}{\sigma_t^3(\T_0)} \eta_t (1-\eta_t^2) \big\{ \pa_\T \mu_t(\T_0) \pa_{\T'} \sigma_t^2(\T_0)+ \pa_\T \sigma_t^2(\T_0) \pa_{\T'} \mu_t(\T_0)\big\} \bigg].
\end{eqnarray*}
Observing that $\E (1-\eta_t^2)^2 =2$ and $\E \eta_t(1-\eta_t^2) =0$ under $H_0$, we have
\begin{eqnarray}\label{epl1}
\E[  \pa_\T l(X_t;\T_0)\pa_{\T'} l(X_t;\T_0) |\mathcal{F}_{t-1}]
=
\frac{1}{2}  \bigg( \frac{1}{\sigma_t^4(\T_0)}  \pa_\T \sigma_t^2(\T_0) \pa_{\T'} \sigma_t^2(\T_0) + \frac{2}{\sigma_t^2(\T_0)} \pa_\T \mu_t(\T_0) \pa_{\T'} \mu_t(\T_0) \bigg).
\end{eqnarray}
In a similar way, one can obtain that 
\begin{eqnarray*}
\paa l(X_t;\T_0)&=&
 -\frac{1}{2} \bigg[ \frac{1}{\sigma_t^2(\T_0)} (1-\eta_t^2) \paa \sigma_t^2(\T_0)
 + \frac{1}{\sigma_t^4(\T_0)} (2\eta_t^2-1) \pa_\T \sigma_t^2(\T_0) \pa_{\T'} \sigma_t^2(\T_0)\\
 &&\qquad +\frac{2}{\sigma_t^3(\T_0)}\eta_t \big\{ \pa_\T \mu_t (\T_0)\pa_{\T'} \sigma_t^2(\T_0)+ \pa_\T \sigma_t^2(\T_0)\pa_{\T'} \mu_t (\T_0)\big\} \\
 &&\qquad +   \frac{2}{\sigma_t^2(\T_0)}\pa_\T \mu_t (\T_0)\pa_{\T'} \mu_t (\T_0)  - \frac{2}{\sigma_t(\T_0)}\eta_t\paa \mu_t (\T_0) \bigg]
\end{eqnarray*}
and thus we have
\begin{eqnarray}\label{epl2}
\E[ \paa l(X_t;\T_0)|\mathcal{F}_{t-1}] =-\frac{1}{2} \bigg(\frac{1}{\sigma_t^4(\T_0)} \pa_\T \sigma_t^2(\T_0) \pa_{\T'} \sigma_t^2(\T_0)   + \frac{2}{\sigma_t^2(\T_0)}\pa_\T \mu_t (\T_0)\pa_{\T'} \mu_t (\T_0) \bigg),
\end{eqnarray}
from which and (\ref{epl1}) we get the lemma.
\end{proof}

\noindent{\bf Proof of Theorem \ref{thm_main}}\\
From (\ref{epl1}) and (\ref{epl2}), we can see that  $\left\{ (d(X_t;\T_0), \mathcal{F}_{t-1})\right\}$ is a martingale difference. Hence, by the CLT for the martingale differences, we have
	\begin{eqnarray*}
D_n(\T_0):=\frac{1}{\sqrt{n}} \sum_{t=1}^n d(X_t;\T_0) \stackrel{d}{\longrightarrow} N_q ({\bf 0}, \Sigma_0),
\end{eqnarray*}
where $\Sigma_0= \cov(d(X;\T_0))$.  By Taylor's theorem, we can write that
\begin{eqnarray}\label{taylor}
D_n(\hat \T_n) =D_n(\T_0)+ \frac{1}{\sqrt{n}}\nabla D_n(\tilde\T_n)\sqrt{n}(\hat \T_n -\T_0),
\end{eqnarray}
where $\nabla D_n$ is the Jacobian matrix of $D_n$ and $\tilde \T_n$ is a point between $\hat \T_n$ and $\T_0$.\\
We first note that since $\tilde \theta_n$ also converges almost surely to $\T_0$, we have by condition {\bf C2} that 
\begin{eqnarray}\label{Jac.D}
\frac{1}{\sqrt{n}}\nabla D_n(\tilde\T_n) =\frac{1}{n} \sum_{t=1}^n \nabla  d(X_t;\tilde\T_n) \stackrel{a.s.}{\longrightarrow}\mathcal{K}:=\E [\nabla d(X_t;\T_0)].
\end{eqnarray}
Further, using (\ref{P1}) and assumption {\bf A7}, we have
\begin{eqnarray}\label{P2}
\sqrt{n} (\hat{\theta}_{n}-\theta_0)
= -\mathcal{J}^{-1}\frac{1}{\sqrt{n}} \sum_{t=1}^n \pa_\T l(X_t ;\T_0)+o_P(1).
\end{eqnarray}
Hence, we can see from (\ref{Jac.D}) and (\ref{P2}) that
\begin{eqnarray*}
\frac{1}{\sqrt{n}}\nabla D_n(\tilde\T_n)\sqrt{n} (\hat{\theta}_{n}-\theta_0)+\mathcal{K}\mathcal{J}^{-1}\frac{1}{\sqrt{n}} \sum_{t=1}^n \pa_\T l(X_t ;\T_0)=o_P(1)
\end{eqnarray*}
and thus, by (\ref{taylor}), we have
\begin{eqnarray*}
 D_n(\hat \T_n)&=&D_n(\T_0)-\mathcal{K}\mathcal{J}^{-1}\frac{1}{\sqrt{n}} \sum_{t=1}^n \pa_\T l(X_t ;\T_0)+o_P(1)\\
 &=& \frac{1}{\sqrt{n}}\sum_{t=1}^n \big(  d(X_t;\T_0) -\mathcal{K}\mathcal{J}^{-1}\pa_\T l(X_t ;\T_0)\big)+o_P(1)
\end{eqnarray*}
Recall from (\ref{epl1}) that $\left\{ (\pa_\T l(X_t;\T_0), \mathcal{F}_{t-1})\right\}$ is a martingale difference, hence $\left\{ (d(X_t;\T_0) -\mathcal{K}\mathcal{J}^{-1}\right. $ $\left.\pa_\T l(X_t ;\T_0),  \mathcal{F}_{t-1})\right\}$ also becomes a martingale difference. Thus, we have by the CLT for martingales that
\begin{eqnarray*}
 D_n(\hat \T_n) \stackrel{d}{\longrightarrow} N_q ( {\bf 0}, \Sigma ),
\end{eqnarray*}
where $ \Sigma= \cov\left( d(X_t;\T_0)- \mathcal{K}\mathcal{J}^{-1}\pa_\T l(X_t ;\T_0) \right) $.

Since $\hat\T_n$ converges almost surely to $\T_0$, we have by assumption {\bf C4} that for sufficiently large $n$,
	\begin{align*}	
		 \frac{1}{\sqrt{n}} \sum_{t=1}^{n}  \left\|  \partial^{2}_{\theta \theta^{'}} \tilde{l} ( X_{t};\hat{\theta}_n ) - \partial^{2}_{\theta \theta^{'}} l ( X_{t};\hat{\theta}_n )    \right\|
		\le \frac{1}{\sqrt{n}}  \sum_{t=1}^{n} \sup_{\theta \in N(\T_0)} \left\| \partial^{2}_{\theta \theta^{'}} \tilde{l} \left( X_{t};\theta \right) - \partial^{2}_{\theta \theta^{'}} l \left( X_{t};\theta \right)  \right\|=o_P(1).
	\end{align*}
and
	\begin{align*}	
		& \frac{1}{\sqrt{n}} \sum_{t=1}^{n}  \left\|  \partial_{\theta } \tilde{l}(X_{t};\hat{\theta}_n ) \partial_{\theta^{'} } \tilde{l}(X_{t};\hat{\theta}_n ) - \partial_{\theta} l(X_{t};\hat{\theta}_{n}) \partial_{\theta^{'}} l(X_{t};\hat{\theta}_{n})   \right\|\\
		&
		\le \frac{1}{\sqrt{n}} \sum_{t=1}^{n} \sup_{\theta \in \mathcal{N}(\theta_{0})}  \left\|  \partial_{\theta } \tilde{l}(X_{t};\T) \partial_{\theta^{'} } \tilde{l}(X_{t};\T) - \partial_{\theta} l(X_{t};\hat{\theta}_{n}) \partial_{\theta^{'}} l(X_{t};\hat{\theta}_{n}) \right\| = o_P(1),
	\end{align*}
which ensure that
\[ \frac{1}{\sqrt{n}} \sum_{t=1}^n \|\tilde d(X_t;\hat \T_n) - d(X_t; \T_0)\| =o_P(1).\]	
This completes the proof.
\hfill{$\Box$}\\

\begin{lm}\label{lm2}
Under $H_0$, it holds that
\[ \E \big[\nabla d(X_t;\T_0)\big]=-\E \big[d(X_t;\T_0) \pa_{\T'} l(X_t;\T_0)\big].\]
\end{lm}
\begin{proof} 
Let us denote $\{X_t^\T\}$ be the process from the model (\ref{model1}) with the parameter $\T$. Then, following the same argument in Lemma \ref{lm1}, one can see that under $H_0$,
\[ \E \big[ \pai l(X_{t}^\T;\theta)\paj l(X_{t}^\T;\theta)|\mathcal{F}_{t-1}\big]= -\E\big[ \paaij l(X_{t}^\T;\theta)|\mathcal{F}_{t-1}\big].\]
Since the conditional distribution of $X_t^\T$ given $\mathcal{F}_{t-1}$ is $N(\mu_t(\T), \sigma_t^2(\T))$, we can express the above as 
\[ \int \pai l(x;\theta)\paj l(x;\theta)f(x;\theta)dx= -\int  \paaij l(x;\theta)f(x;\theta)dx,\]
where $f(x;\T)$ is the pdf of $N(\mu_t(\T), \sigma_t^2(\T))$. Differentiating the both sides of the above with respect to $\T_l$, we obtain 
\begin{eqnarray*}
&&\E \big[ \paail l(X_{t}^\T;\theta)\paj l(X_{t}^\T;\theta)+\pai l(X_{t}^\T;\theta)\paajl l(X_{t}^\T;\theta)+\pai l(X_{t}^\T;\theta)\paj l(X_{t}^\T;\theta)\pal l(X_{t}^\T;\theta)|\mathcal{F}_{t-1}\big]\\
&&=- \E \big[ \paaa l(X_{t}^\T;\theta)+\paaij l(X_{t}^\T;\theta)\pal l(X_{t}^\T;\theta)|\mathcal{F}_{t-1}\big],
\end{eqnarray*}
from which we can see that 
\begin{eqnarray*}
&&\E \big[ \pal d_k(X_t;\T_0) |\mathcal{F}_{t-1}\big]\\
&&=\E \big[ \paaa l(X_{t};\theta_0)+ \paail l(X_{t};\theta_0)\paj l(X_{t};\theta_0)+\pai l(X_{t};\theta_0)\paajl l(X_{t};\theta_0)|\mathcal{F}_{t-1}\big]\\
&&=- \E \big[\big\{\paaij l(X_{t};\theta_0)+\pai l(X_{t};\theta_0)\paj l(X_{t};\theta_0)\big\}\pal l(X_{t};\theta_0)|\mathcal{F}_{t-1}\big]\\
&&= -\E \big[ d_k(X_t;\T_0) \pal l(X_{t};\theta_0)|\mathcal{F}_{t-1}\big].
\end{eqnarray*}
This asserts the lemma. 
\end{proof}

\subsection{Proofs for Subsection \ref{Sec:3_1}}
$\tilde \ep_t$ and $\ep_t$ that will be shown in Lemmas \ref{lm_TMA0}-\ref{lm_approx} are the ones defined in (\ref{tma.tep}) and (\ref{tma.ep}), respectively, and $\Theta$ is the parameter space given in (\ref{tma.ps}).
\begin{lm} \label{lm_TMA0}
Under $H_0$, we have that  for all $d\geq1$,
    \[\E \sup_{\T\in\Theta} |\ep_t|^d<\infty,\quad \E \sup_{\T\in\Theta} | \pa_{\T_i} \ep_t|^d <\infty
    ,\quad \E \sup_{\T\in\Theta} | \pa^2_{\T_i\T_j} \ep_t |^d <\infty,\quad 
    \E \sup_{\T\in\Theta} | \pa^3_{\T_i\T_j\T_k} \ep_t |^d <\infty.\]
\end{lm}
\begin{proof}
By the boundedness of $\Theta$, we have
\begin{eqnarray*}
	|X_{t}| \leq  \big| \big(\phi_0 + \xi_0 I (X_{t-1} \le u) \big) \sigma \ep_{t-1}| + | \sigma \ep_{t}| \les  |\ep_{t-1}| +|\ep_t|.
\end{eqnarray*}
Since $\ep_t$ follows the normal distribution under $H_0$, we can see that $X_t$ admits moments of any order. Now letting $A_t(\T):=A_t=-\left(\phi +\xi I (X_t \le u) \right)$, it can be written that
\begin{eqnarray}\label{ept}
	\ep_t= \frac{1}{\sigma}  X_{t} +\frac{1}{\sigma} \sum_{j=1}^{\infty} \Big( \prod_{i=1}^{j} A_{t-i} \Big) X_{t-j}.
\end{eqnarray}
Noting that $|A_t|\leq |\phi|\vee|\phi+\xi|\leq c_1$, we have that for any $\rho \in [c_1,1)$,  $\prod_{i=1}^{j} |A_{t-i}| \leq \rho^j$. Hence, it follows from (\ref{ept}) and Minkowski's inequality that 
\begin{eqnarray}\label{mnt}
    \sup_{\T\in\Theta}\|\ep_t\|_d\les \|X_t\|_d + \sum_{j=1}^\infty \rho^j \|X_{t-j}\|_d <\infty,
\end{eqnarray}
where $\|\cdot\|_d$ is the $L_d$-norm, and consequently we have $\E \sup_{\T\in\Theta} |\ep_t|^d<\infty$.\\
By simple algebra, we have that 
\begin{eqnarray*}
&&\frac{\pa\ep_t}{\pa \phi}= -\frac{1}{\sigma} \sum_{j=1}^\infty \sum_{k=1}^j \Big(\prod_{i=1 ,i\neq k}^j A_{t-i}\Big) X_{t-j},
\quad\frac{\pa\ep_t}{\pa \xi}= -\frac{1}{\sigma} \sum_{j=1}^\infty \sum_{k=1}^j I(X_{t-k} \leq u) \Big(\prod_{i=1 ,i\neq k}^j A_{t-i}\Big) X_{t-j},\\
&& \frac{\pa\ep_t}{\pa \sigma^2}= -\frac{1}{2 \sigma^3}\Big( X_t +\sum_{j=1}^{\infty} \Big( \prod_{i=1}^{j} A_{t-i} \Big) X_{t-j}\Big),
\end{eqnarray*}
and
\begin{eqnarray*}
    &&\frac{\pa^2 \ep_t}{\pa \phi^2} = \frac{2}{\sigma} \sum_{j=1}^{\infty} \sum_{k=1}^{j} \sum_{l=1,l\neq k}^j \Big( \prod_{i=1,i\neq k,l}^{j} A_{t-i}\Big) X_{t-j},\\
	&&\frac{\pa^2 \ep_t}{\pa \xi^2} = \frac{2}{\sigma} \sum_{j=1}^{\infty} \sum_{k=1}^{j}I(X_{t-k} \leq u) \sum_{l=1,l\neq k}^j I(X_{t-l} \leq u)\Big( \prod_{i=1,i\neq k,l}^{j} A_{t-i}\Big) X_{t-j},\\
	&&\frac{\pa^2 \ep_t}{\pa \sigma^4} = \frac{3}{4 \sigma^5} \Big( X_t + \sum_{j=1}^{\infty} \Big( \prod_{i=1}^{j} A_{t-i} \Big) X_{t-j} \Big),\\
	&&\frac{\pa^2 \ep_t}{\pa \phi\pa\xi} = \frac{2}{\sigma} \sum_{j=1}^{\infty} \sum_{k=1}^{j} \sum_{l=1,l\neq k}^j I(X_{t-l} \leq u)\Big( \prod_{i=1,i\neq k,l}^{j} A_{t-i}\Big) X_{t-j},\\
	&&\frac{\pa^2 \ep_t}{\pa \phi\pa\sigma^2} = \frac{1}{2\sigma^3} \sum_{j=1}^\infty \sum_{k=1}^j \Big(\prod_{i=1 ,i\neq k}^j A_{t-i}\Big) X_{t-j},\quad
	\frac{\pa^2 \ep_t}{\pa \xi \pa \sigma^2} = \frac{1}{2\sigma^3} \sum_{j=1}^\infty \sum_{k=1}^j I(X_{t-k} \leq u) \Big(\prod_{i=1 ,i\neq k}^j A_{t-i}\Big) X_{t-j}.
\end{eqnarray*}
Similarly to (\ref{mnt}), one can show that for $1\leq i,j\leq 3$,
\begin{eqnarray*}
    && \sup_{\T\in\Theta} \|\pa_{\T_i} \ep_t\|_d\vee  \sup_{\T\in\Theta} \|\pa^2_{\T_i\T_j} \ep_t\|_d \\
&&\les \|X_t\|_d+ \sum_{j=1}^\infty \rho^{j} \|X_{t-j}\|_d+ \sum_{j=1}^\infty j \rho^{j-1} \|X_{t-j}\|_d
+\sum_{j=1}^\infty j(j-1)\rho^{j-2} \|X_{t-j}\|_d <\infty.
\end{eqnarray*} 
The moment condition for the third derivatives can also be shown in the same way and we omit the proof for brevity.  
\end{proof}

\begin{lm} \label{lm_TMA1}
It holds that for some $\rho \in (0,1)$,
    \[\sup_{\T\in\Theta} |\tilde{\ep}_t - \ep_t| \les \rho^t,
    \quad  \sup_{\T\in\Theta}|\pa_{\T_i} \tilde{\ep}_t - \pa_{\T_t} \ep_t |\les (1+t)\rho^t, \quad
    \sup_{\T\in\Theta} |\pa^2_{\T_i\T_j} \tilde{\ep}_t - \pa^2_{\T_i\T_j} \ep_t| \les (1+t+t^2)\rho^t.\]
\end{lm}
\begin{proof}
Note that
\begin{eqnarray*}
	\tilde{\ep}_t= \frac{1}{\sigma}  X_{t} +\frac{1}{\sigma} \sum_{j=1}^{t-1} \Big( \prod_{i=1}^{j} A_{t-i} \Big) X_{t-j}.
\end{eqnarray*}
From (\ref{ept}), we have
\begin{eqnarray*}
	| \tilde{\ep}_{t} - \ep_{t}  |&=&  \frac{1}{\sigma}   \Big| \sum_{j=1}^{t-1} \Big( \prod_{i=1}^{j} A_{t-i} \Big) X_{t-j} - \sum_{j=1}^{\infty} \Big( \prod_{i=1}^{j} A_{t-i} \Big) X_{t-j}  \Big| \\
	&\les&   \Big| \sum_{j=t}^{\infty} \Big( \prod_{i=1}^{j} A_{t-i} \Big) X_{t-j} \Big|\\
	&\les& \rho^t \sum_{j=t}^\infty \rho^{j-t} |X_{t-j}| =\rho^t \sum_{j=0}^\infty \rho^{j} |X_{-j}|\quad a.s.
\end{eqnarray*}
Since $\sum_{j=0}^\infty \rho^{j}\E|X_{-j}|<\infty$, $\sum_{j=0}^\infty \rho^{j}|X_{-j}|$  is well defined. Thus, we have $\sup_{\T\in\Theta}| \tilde{\ep}_t - \ep_t | \les \rho^t$ a.s.\\
The first and second derivatives of $\tilde \ep_t$ can be obtained similarly to those of $\ep_t$ in Lemma \ref{lm_TMA0}. In the same fashion as above, we can see that
\begin{eqnarray*}
	\sup_{\T\in\Theta}\left| \pa_\phi \tilde \ep_t - \pa_\phi \ep_t \right|
	&\les& 	\sup_{\T\in\Theta}\Big| \sum_{j=1}^\infty \sum_{k=1}^j \Big(\prod_{i=1 ,i\neq k}^j A_{t-i}\Big) X_{t-j}- \sum_{j=1}^{t-1} \sum_{k=1}^j \Big(\prod_{i=1 ,i\neq k}^j A_{t-i}\Big) X_{t-j} \Big|\\
	&=&	\sup_{\T\in\Theta}\Big| \sum_{j=t}^\infty \sum_{k=1}^j \Big(\prod_{i=1 ,i\neq k}^j A_{t-i}\Big) X_{t-j}\Big|\\
    &\les& \rho^t\sum_{j=t}^\infty j\rho^{j-1-t} |X_{t-j}|\\
    &=&\rho^t\Big(\sum_{j=t}^\infty (j-t)\rho^{j-1-t} |X_{t-j}|+t\sum_{j=t}^\infty \rho^{j-1-t} |X_{t-j}|\Big)\\
    &=&\rho^t\Big(\sum_{j=0}^\infty j\rho^{j-1} |X_{-j}|+t\sum_{j=0}^\infty \rho^{j-1} |X_{-j}|\Big)\les(1+t)\rho^t\quad a.s.
\end{eqnarray*}
Similarly, we can show that 
\[	\sup_{\T\in\Theta}\left| \pa_\xi \tilde \ep_t - \pa_\xi \ep_t \right| \les (1+t)\rho^t\quad a.s.,\qquad	\sup_{\T\in\Theta}\left| \pa_{\sigma^2} \tilde \ep_t - \pa_{\sigma^2} \ep_t \right| \les \rho^t\quad a.s.\]
and
\[	\sup_{\T\in\Theta}\big\{\left| \pa^2_{\phi\phi} \tilde \ep_t - \pa^2_{\phi\phi} \ep_t \right|\vee\left| \pa^2_{\xi\xi} \tilde \ep_t - \pa^2_{\xi\xi} \ep_t \right|
\vee\left| \pa^2_{\phi\xi} \tilde \ep_t - \pa^2_{\phi\xi} \ep_t \right|\big\}\les (1+t+t^2)\rho^t\quad a.s.\]
\[	\sup_{\T\in\Theta}\big\{\left| \pa^2_{\phi\sigma^2} \tilde \ep_t - \pa^2_{\phi\sigma^2} \ep_t \right| \vee \left| \pa^2_{\xi\sigma^2} \tilde \ep_t - \pa^2_{\xi\sigma^2} \ep_t \right|\big\}\les (1+t)\rho^t\quad a.s.,\qquad	\sup_{\T\in\Theta}\left| \pa^2_{\sigma^2\sigma^2} \tilde \ep_t - \pa^2_{\sigma^2\sigma^2} \ep_t \right| \les \rho^t\quad a.s.,\]
which yield the last two equalities in the lemma.
\end{proof}

\begin{lm}\label{moment2} Under $H_0$, we have that for all $d\geq1$,
     \[\E \sup_{\theta \in \Theta }\big|\pa_{\T_i} l_t(\T)\big|^d  <\infty, \quad\E \sup_{\theta \in \Theta }\big|\pa^2_{\T_i\T_j} l_t(\T)\big|^d  <\infty,
     \quad\E \sup_{\theta \in \Theta }\big|\pa^3_{\T_i\T_j\T_k} l_t(\T)\big|^d  <\infty.\]
\end{lm}
\begin{proof}
Note that  $|\pa_{\T_i} \sigma^2|\leq 1$ and $\pa^2_{\T_i\T_j} \sigma^2=0$.
Since $l_t(\T)= -\frac{1}{2} \log \sigma^2 -\frac{1}{2}\ep_t^2$, we have
\begin{eqnarray*}
|\pa_{\T_i} l_t(\T)| &=&\frac{1}{2}\Big|\frac{1}{\sigma^2}\pa_{\T_i} \sigma^2+2\ep_t\pa_{\T_i}\ep_{t}\Big| \les 1+|\ep_t||\pa_{\T_i}\ep_t|\\
|\pa^2_{\T_i\T_j} l_t(\T)| &=& \frac{1}{2}\Big| \frac{1}{\sigma^4}\pa_{\T_i} \sigma^2 \pa_{\T_j} \sigma^2 -\frac{1}{\sigma^2} \pa^2_{\T_i\T_j} \sigma^2 -2 \pa_{\T_i} \ep_t \pa_{\T_j} \ep_t-2\ep_t   \pa^2_{\T_i\T_j} \ep_t \Big|\\
 &\les& 1+ |
 \pa_{\T_i} \ep_t| | \pa_{\T_j} \ep_t|+ |\ep_t| |\pa^2_{\T_i\T_j} \ep_t |
\end{eqnarray*}
and
\begin{eqnarray*}
|\pa^3_{\T_i\T_j\T_k} l_t(\T)| &=& \Big| \frac{1}{\sigma^6}\pa_{\T_i} \sigma^2 \pa_{\T_j} \sigma^2\pa_{\T_k} \sigma^2 + \pa^2_{\T_i\T_k} \ep_t \pa_{\T_j} \ep_t+\pa_{\T_i} \ep_t \pa^2_{\T_j\T_k}\ep_t+\pa_{\T_k}\ep_t  \pa^2_{\T_i\T_j} \ep_t
+\ep_t  \pa^3_{\T_i\T_j\T_k} \ep_t\Big|\\
 &\les& 1+ |\pa^2_{\T_i\T_k} \ep_t| | \pa_{\T_j} \ep_t|+ |\pa_{\T_i} \ep_t| | \pa^2_{\T_j\T_k} \ep_t|+|\pa_{\T_k} \ep_t| | \pa^2_{\T_i\T_j} \ep_t|+|\ep_t| |\pa^3_{\T_i\T_j\T_k} \ep_t |
\end{eqnarray*}
which together with Lemma \ref{lm_TMA0} and the Cauchy–Schwarz inequality yields the lemma.  
\end{proof}

\begin{lm}\label{nonsing} Under $H_0$, $\E \big[\paa l_t(\T_0) \big]$ is a nonsingular matrix.
\end{lm}
\begin{proof} 
By Lemma \ref{moment2}, $\E \big[\paa l_t(\T_0) \big]$ exists and it is negative semidefinite since $\E\left[ \partial^{2}_{\theta \theta'} l(X_{t};\theta_{0})\right]= -\cov \big(\pa_{\T}l_t(\theta_{0})\big)$ by Lemma \ref{lm1}. Assume that for some $z =(z_1,z_2,z_3)' \in \mathbb{R}^3$, $z' \E[ \paa l_t(\T_0)] z =0$.
Then, it follows from (\ref{epl2}) that
\begin{eqnarray*}
   z'\E\big[ \paa l_t(\T_0)\big]z &=&-\frac{1}{2}z'\E\bigg[\frac{1}{\sigma_t^4(\T_0)} \pa_\T \sigma_t^2(\T_0) \pa_{\T'} \sigma_t^2(\T_0)   + \frac{2}{\sigma_t^2(\T_0)}\pa_\T \mu_t (\T_0)\pa_{\T'} \mu_t (\T_0) \bigg]z\\
    &=& -\frac{1}{2} \E\bigg[\frac{1}{\sigma_t^4(\T_0)} \big(z'\pa_\T \sigma_t^2(\T_0)\big)^2   + \frac{2}{\sigma_t^2(\T_0)}\big(z'\pa_\T \mu_t (\T_0)\big)^2 \bigg]=0.
\end{eqnarray*}
Hence, we can see that $z'\pa_\T \sigma_t^2(\T_0)$ and $z'\pa_\T \mu_t (\T_0)$ are equal to zero almost surely. Noting that $\pa_\T \sigma_t^2(\T_0)=(0,0,1)'$, we have $z_3=0$. From the second equation, it should also hold that $(z_1+I(X_{t-1} \leq u)z_2)\eta_{t-1}=0$ almost surely, which implies $z_1=z_2=0$. Therefore, $\E \big[\paa l_t(\T_0) \big]$ is invertible.
\end{proof}

\begin{lm} \label{lm_approx}
	Under $H_0$, we have
\begin{eqnarray*}
    &&\sum_{t=1}^n \sup_{\theta \in \Theta}\big\|\pa_{\theta}\,l_t(\T)-\pa_{\theta}\,\tilde l_t(\T)\big\|= O(1)\quad a.s.\\
	&&
	\sum_{t=1}^{n} \sup_{\theta \in \Theta}  \big\|  \partial_{\theta} l_t(\theta) \partial_{\theta^{'}} l_t(\theta)- \partial_{\theta } \tilde l_t(\theta) \partial_{\theta^{'} } \tilde l_t(\theta) \big\| = O(1)\quad a.s. \\
	&&\sum_{t=1}^{n}\sup_{\theta \in \Theta}\big\|\paa  l_t(\theta)
	-\paa  \tilde l_t(\theta)\big\|= O(1)\quad a.s.
\end{eqnarray*}
\end{lm}	
	
\begin{proof}
Using Lemma \ref{lm_TMA1} and $|\tilde \ep_t| \les |\ep_t|+\rho^t \les 1+|\ep_t|$, we have
\begin{eqnarray*}
   \big|\pa_{\T_i}l_t(\T)-\pa_{\T_i}\tilde l_t(\T)\big|&=& | \tilde \ep_t \pa_{\T_i} \tilde \ep_t - \ep_t \pa_{\T_i} \ep_t |\\
    &\leq& |\tilde \ep_t| | \pa_{\T_i} \tilde \ep_t- \pa_{\T_i} \ep_t| +
    |\tilde \ep_t -\ep_t| |\pa_{\T_i} \ep_t|\\
    &\les& (1+t)(1+|\ep_t|+|\pa_{\T_i} \ep_t|)\rho^t:=P_{t,i}(\T)\rho^t.
\end{eqnarray*}
Also noting that
$| \pa_{\T_i} l_t(\T) | \les 1+ |\ep_t| |\pa_{\T_i} \ep_t|$ and 
$ | \pa_{\T_j} \tilde l_t(\T) | \les |\pa_{\T_j} l_t(\T)|+P_{t,j}(\T)$, we have
\begin{eqnarray*}
\big|\pa_{\T_i} l_t(\T) \pa_{\T_j} l_t(\T)-\pa_{\T_i}\tilde l_t(\T)
\pa_{\T_j}\tilde l_t(\T)\big| &\leq& \big|\pa_{\T_i} l_t(\T)\big|\big|\pa_{\T_j} l_t(\T)-\pa_{\T_j} \tilde l_t(\T)\big|+\big|\pa_{\T_j} \tilde l_t(\T)\big|\big|\pa_{\T_i} l_t(\T)-\pa_{\T_i} \tilde l_t(\T)\big|\\
&\les& \big(1+ |\ep_t| |\pa_{\T_i} \ep_t|+|\ep_t| |\pa_{\T_j} \ep_t|+P_{t,j}(\T)\big) \big(P_{t,i}(\T)+P_{t,j}(\T)\big)\rho^t\\
&:=& Q_{t,i,j}(\T)\rho^t.
\end{eqnarray*}
In view of the moment results in Lemma \ref{lm_TMA0} and the Cauchy-Schwarz inequality, we can see that 
\begin{eqnarray*}
    \E \sup_{\T \in \Theta} P_{t,i}(\T) < \infty, \quad  \E \sup_{\T \in \Theta} Q_{t,i,j}(\T) < \infty.
\end{eqnarray*}
And thus, the first two results in the lemma follow from 
\begin{eqnarray*}
     \sum_{t=1}^{\infty} \rho^t \E \sup_{\T \in \Theta} P_{t,i}(\T) < \infty, \quad  \sum_{t=1}^{\infty} \rho^t \E \sup_{\T \in \Theta} Q_{t,i,j}(\T) < \infty,
\end{eqnarray*}
respectively. Since the last result can be shown in a similar way, we omit its proof for brevity.
\end{proof}	

\subsection{Proofs for Subsection \ref{Sec:3_2}}
\noindent To establish Lemmas \ref{lm_garch_C2} and \ref{lm_garch_C4} below, we shall use some technical results obtained by \cite{francq:zakoian:2004}. Under the assumptions $\bf{G1}$-$\bf{G4}$ and $H_0$, the followings hold: for any $d\geq1$,
\begin{eqnarray}\label{tech1}
   \E \sup_{\T\in\Theta^*} \Big|\frac{1}{\sigma^2_t} \pa_{\T_k} \sigma^2_t \Big|^d <\infty,\quad  \E \sup_{\T\in\Theta^*} \Big|\frac{1}{\sigma^2_t} \pa^2_{\T_i \T_j} \sigma^2_t \Big|^d <\infty,
 \quad
 \E \sup_{\T\in\Theta^*} \Big|\frac{1}{\sigma^2_t} \pa^3_{\T_i \T_j \T_k} \sigma^2_t \Big|^d <\infty,
\end{eqnarray} 
 where $\Theta^* \subset \Theta^o$ is a compact set containing $\T_0$, and for some constant $\rho \in (0,1)$,
     \begin{eqnarray} \label{tech2}
        \sup_{\T \in \Theta} \left\{  \big| \pa_{\T_i} \sigma^2_t - \pa_{\T_i} \tilde{\sigma}^2_t\big| \vee \big| \pa^2_{\T_i\T_j} \sigma^2_t - \pa^2_{\T_i\T_j} \tilde{\sigma}^2_t \big| \right\}  \les  \rho^t \ a.s., \quad \left| \frac{1}{\sigma^2_t} - \frac{1}{\tilde{\sigma}^2_t} \right| \les \frac{\rho^t}{\sigma^2_t}\ a.s.
    \end{eqnarray}
 From (\ref{tech2}), it can be shown that
 \begin{eqnarray} \label{tech-deriv}
 \left| \frac{1}{\tilde{\sigma}^2_t} \frac{\pa \tilde{\sigma}^2_t}{\pa \T_i} \right| \les 1+ \left| \frac{1}{\sigma^2_t} \frac{\pa \sigma^2_t}{\pa \T_i} \right|, 
 \quad 
 \frac{X^2_t}{\tilde{\sigma}^2_t} \les \frac{X^2_t}{\sigma^2_t}.
 \end{eqnarray}

\begin{lm} \label{lm_garch_C2}
    Suppose that $\bf{G1}$-$\bf{G4}$ hold. Then, under $H_0$, we have that for any $d\geq1$,
       \[\E \sup_{\theta \in  N(\T_0) }\big|\pa_{\T_i} l_t(\T)\big|^d  <\infty, \quad\E \sup_{\theta \in  N(\T_0) }\big|\pa^2_{\T_i\T_j} l_t(\T)\big|^d  <\infty,
     \quad\E \sup_{\theta \in  N(\T_0) }\big|\pa^3_{\T_i\T_j\T_k} l_t(\T)\big|^d  <\infty,\]
where $N(\theta_{0})$ is a neighborhood of $\T_0$.
\end{lm}
\begin{proof}
By Lemma 1 in \cite{lee2009minimum}, we can take a neighborhood $N(\T_0)$ included in $\Theta^*$ such that  for any $d\geq 1$,
 \begin{eqnarray}\label{tech3}
 \E \sup_{\T \in N(\T_0)} \frac{X^{2d}_t}{\sigma^{2d}_t}  <\infty.
 \end{eqnarray}
Observe that 
 \begin{eqnarray}
 \big|\pa_{\T_i} l_t(\T)\big| &=& \bigg|\Big( 1-\frac{X^2_t}{\sigma^2_t} \Big)  \frac{1}{\sigma^2_t} \pa_{\T_i} \sigma^2_t\bigg| 
 \les \bigg| 1+\frac{X_t^2}{\sigma_t^2} \bigg| \bigg|\frac{1}{\sigma^2_t} \pa_{\T_i} \sigma^2_t\bigg| \label{upbd.pa.l}\\
 \big|\pa^2_{\T_i\T_j} l_t(\T)\big| &=& \bigg|\Big( 1-\frac{X^2_t}{\sigma^2_t} \Big) \frac{1}{\sigma^2_t} \pa^2_{\T_i \T_j} \sigma^2_t + \Big( 2 \frac{X^2_t}{\sigma^2_t} -1 \Big)  \frac{1}{\sigma^2_t} \pa_{\T_i} \sigma^2_t  \frac{1}{\sigma^2_t} \pa_{\T_j} \sigma^2_t\bigg| \nonumber\\
 &\les& 
 \bigg| 1+\frac{X_t^2}{\sigma_t^2} \bigg| \bigg( \bigg|\frac{1}{\sigma^2_t} \pa^2_{\T_i \T_j} \sigma^2_t\bigg| + \bigg|\frac{1}{\sigma^2_t} \pa_{\T_i} \sigma^2_t  \frac{1}{\sigma^2_t} \pa_{\T_j} \sigma^2_t\bigg|\bigg) \nonumber
 \end{eqnarray}
 and
\begin{eqnarray*}
	\big|\pa^3_{\T_i \T_j \T_k} l_t(\T)\big| &=& \bigg|\Big( 1-\frac{X^2_t}{\sigma^2_t} \Big)  \frac{1}{\sigma^2_t} \pa^3_{\T_i \T_j \T_k} \sigma^2_t 
	+ \Big( 2-6\frac{X^2_t}{\sigma^2_t} \Big)   \frac{1}{\sigma^2_t} \pa^3_{\T_i \T_j \T_k} \sigma^2_t \\
	 &&+\Big( 2\frac{X^2_t}{\sigma^2_t}-1 \Big)  \Big(\frac{1}{\sigma^2_t} \pa_{\T_i} \sigma^2_t    \frac{1}{\sigma^2_t} \pa^2_{\T_j \T_k} \sigma^2_t  +  \frac{1}{\sigma^2_t}\pa_{\T_j} \sigma^2_t   \frac{1}{\sigma^2_t} \pa^2_{\T_i \T_k} \sigma^2_t+ \frac{1}{\sigma^2_t} \pa_{\T_k} \sigma^2_t \frac{1}{\sigma^2_t} \pa^2_{\T_i \T_j} \sigma^2_t \Big) \bigg|\\
	&\les& 
	\bigg| 1+\frac{X_t^2}{\sigma_t^2} \bigg(
	\bigg| \frac{1}{\sigma^2_t} \pa^3_{\T_i \T_j \T_k} \sigma^2_t \bigg|
	+ \bigg|\frac{1}{\sigma^2_t} \pa^3_{\T_i \T_j \T_k} \sigma^2_t \bigg|\\
	&& \hspace{1.7cm}
	+\bigg|\frac{1}{\sigma^2_t} \pa_{\T_i} \sigma^2_t    \frac{1}{\sigma^2_t} \pa^2_{\T_j \T_k} \sigma^2_t\bigg|  +  \bigg|\frac{1}{\sigma^2_t}\pa_{\T_j} \sigma^2_t   \frac{1}{\sigma^2_t} \pa^2_{\T_i \T_k} \sigma^2_t\bigg|+ \bigg|\frac{1}{\sigma^2_t} \pa_{\T_k} \sigma^2_t \frac{1}{\sigma^2_t} \pa^2_{\T_i \T_j} \sigma^2_t\bigg|\bigg).
\end{eqnarray*}
Then, using  (\ref{tech1}),(\ref{tech3}), and the Cauchy–Schwarz inequality, one can establish the lemma.
\end{proof}

\begin{lm} \label{lm_garch_C4}
	Suppose that $\bf{G1}$-$\bf{G4}$ hold. Then, under $H_0$, we have
\[ \sum_{t=1}^{n} \sup_{\theta \in N(\theta_{0})}  \big\|  \partial_{\theta} l_t(\theta) \partial_{\theta^{'}} l_t(\theta)- \partial_{\theta } \tilde l_t(\theta) \partial_{\theta^{'} } \tilde l_t(\theta) \big\| = O(1)\quad a.s.\]
and
\[\sum_{t=1}^{n} \sup_{\theta \in N(\theta_{0})} \big\|\paa  l_t(\theta)
	-\paa  \tilde l_t(\theta)\big\|= O(1)\quad a.s. \]
\end{lm}

\begin{proof}
The lemma can be shown in the same fashion as in Lemma \ref{lm_approx}. Using (\ref{tech2}) and (\ref{tech-deriv}), we have
\begin{eqnarray} \label{temp-1}
	\big| \pa_{\T_i} l_t(\T) - \pa_{\T_i} \tilde{l}_t(\T) \big| &=& \bigg|  \Big(\frac{X^2_t}{\tilde{\sigma}^2_t} - \frac{X^2_t}{\sigma^2_t} \Big)  \frac{1}{\sigma^2_t} \pa_{\T_i} \sigma^2_t + \Big( 1- \frac{X^2_t}{\tilde{\sigma}^2_t} \Big) \Big( \frac{1}{\sigma^2_t} - \frac{1}{\tilde{\sigma}^2_t} \Big) \pa_{\T_i} \sigma^2_t \nonumber \\  \nonumber
	&& + \Big( 1-\frac{X^2_t}{\tilde{\sigma}^2_t} \Big)  \frac{1}{\tilde{\sigma}^2_t} \big(\pa_{\T_i} \sigma^2_t - \pa_{\T_i} \tilde{\sigma}^2_t \big)
	 \bigg|\\ 
	 &\les&  \left( 1+ \frac{X^2_t}{\sigma^2_t} \right) \bigg( 1+ \bigg|\frac{1}{\sigma^2_t} \pa_{\T_i} \sigma^2_t\bigg|\bigg)\rho^t :=  P_{t,i}(\T)\rho^t. 
\end{eqnarray}
and thus it follows from (\ref{upbd.pa.l}) that
\begin{eqnarray*}
&&\big|\pa_{\T_i} l_t(\T) \pa_{\T_j} l_t(\T)-\pa_{\T_i}\tilde l_t(\T)
\pa_{\T_j}\tilde l_t(\T)\big| \\
&&\les 
 \bigg| 1+\frac{X_t^2}{\sigma_t^2} \bigg| \bigg(\bigg|\frac{1}{\sigma^2_t} \pa_{\T_i} \sigma^2_t\bigg|+\bigg|\frac{1}{\sigma^2_t} \pa_{\T_j} \sigma^2_t\bigg|+P_{t,j}(\T)\bigg) \big(P_{t,i}(\T)+P_{t,j}(\T)\big)\rho^t
:= Q_{t,i,j}(\T) \rho^t.
\end{eqnarray*}
By simple algebra with (\ref{tech2}) and (\ref{tech-deriv}), we can also have 
\begin{eqnarray*}
\big|\pa^2_{\T_i \T_j} l_t(\T)-\pa^2_{\T_i \T_j} \tilde l_t(\T)\big|
&\les& \left( 1+\frac{X^2_t}{\sigma^2_t} \right) \left( 1 + \left| \frac{1}{\sigma^2_t} \pa^2_{\T_i \T_j} \sigma^2_t \right| + \left| \frac{1}{\sigma^2_t} \pa_{\T_i} \sigma^2_t \right|  \left| \frac{1}{\sigma^2_t} \pa_{\T_j} \sigma^2_t \right| \right)\rho^t :=   R_{t,i,j}(\T) \rho^t. 
\end{eqnarray*}
Using the moments in (\ref{tech1}) and (\ref{tech3}), we have
\begin{eqnarray*}
    \E \sup_{\T \in \Theta} Q_{t,i,j}(\T) < \infty\quad \mbox{and} \quad  \E \sup_{\T \in \Theta} R_{t,i,j}(\T) < \infty,
\end{eqnarray*}
which assert the lemma. 
\end{proof}

\subsection{Proofs for Subsection \ref{Sec:3_3}}
\begin{lm}\label{lm_dar_C2}    
	Under $H_0$, we have that for all $d\ge1$,
	\[\E \sup_{\theta \in \Theta }\big|\pa_{\T_i} l_t(\T)\big|^d  <\infty, \quad\E \sup_{\theta \in \Theta }\big|\pa^2_{\T_i\T_j} l_t(\T)\big|^d  <\infty,
	\quad\E \sup_{\theta \in \Theta }\big|\pa^3_{\T_i\T_j\T_k} l_t(\T)\big|^d  <\infty.\]
\end{lm}
\begin{proof}
Let $\Lambda_t(\T) =(\phi_0-\phi)X_{t-1}+\epsilon_t   $ and $\Upsilon_t^2(\T)=\omega+\alpha X_{t-1}^2$, where $\ep_t=e_t\sqrt{\omega_0 +\A_0X^2_{t-1}}$ and $e_t\sim N(0,1)$. Then we can write that
	\[ l_t(\T)=-\frac{1}{2}\log\Upsilon_t^2(\T)-\frac{1}{2}\frac{\Lambda_t^2(\T)}{\Upsilon_t^2(\T)}. \]
Note that $\pa_{\T}\Lambda(\T)=(-X_{t-1},0,0)', \pa_{\T}\Upsilon_t^2(\T)=(0,1,X_{t-1}^2)',$ and 
\begin{eqnarray}\label{bound-1}
	\left| \frac{1}{\Upsilon_t^2(\T)}\pa_{\T_i}\Upsilon_t^2(\T) \right| \le \frac{1}{c_2+c_4 X_{t-1}^2}\left(X_{t-1}^2+1\right)\les 1,
\end{eqnarray}
where $c_2$ and $c_4$ are the ones given in (\ref{dar_par}).
Further, we also have from (\ref{bound-1}) that
\begin{eqnarray}\label{bound-2}
      \left|\frac{1}{\Upsilon_t^2(\T)} \frac{\Lambda_t^2(\T)}{\Upsilon_t^2(\T)}\pa_{\T_i}\Upsilon_t^2(\T) \right|
        &\les &
        \left|\frac{\Lambda_t^2(\T)}{\Upsilon_t^2(\T)} \right| \nonumber\\
        &\les&
        \sup_{\T \in \Theta}\left|\frac{\left( \phi-\phi_0 \right)^2X_{t-1}^2}{\omega+\alpha X_{t-1}^2} \right| 
        +
        \sup_{\T \in \Theta}\left|\frac{\omega_0 +\A_0X^2_{t-1}}{\omega+\alpha X_{t-1}^2}\,   e_t^2   \right| \\
      &\les& 1+ e_t^2\nonumber
\end{eqnarray}
and, for any fixed $n \in \mathbb{N}$,
\begin{eqnarray}\label{bound-3}
	\left| \frac{\Lambda_t(\T)}{\Upsilon_t^2(\T)}\pa_{\T_i}\Lambda_t(\T) \right|^{2n} 
 &\le& \sup_{\T \in \Theta}\left| \frac{(\phi_0-\phi)X_{t-1}+\epsilon_t}{\omega+\alpha X_{t-1}^2}X_{t-1} \right|^{2n} \nonumber \\
		&\les& 
  \sup_{\T \in \Theta}\left| \frac{(\phi-\phi_0)^{2n}X_{t-1}^{4n}}{(\omega+\alpha X_{t-1}^2)^{2n}} \right|
  +\sup_{\T \in \Theta}\left| \frac{(\omega_0+\alpha_0X_{t-1}^2)^n e_t^{2n}}{(\omega+\alpha X_{t-1}^2)^{2n}}X_{t-1}^{2n} \right|\\
		&\les&1+e_t^{2n}. \nonumber
\end{eqnarray}
Using (\ref{bound-1}) - (\ref{bound-3}), one can have that
\begin{eqnarray}\label{p-l}
		\big|\pa_{\T_i}l_t(\T)\big|^{2n}=\left|\frac{1}{2}\frac{1}{\Upsilon_t^2(\T)}\left[ \left( 1-\frac{\Lambda_t^2(\T)}{\Upsilon_t^2(\T)} \right)\pa_{\T_i}\Upsilon_t^2(\T)+2\Lambda_t(\T)\pa_{\T_i}\Lambda_t(\T) \right]\right|^{2n} \les 1+e_t^{2n}+e_t^{4n}.
\end{eqnarray}
Similarly to the above, it can also be shown that 
\[ \left| \pa_{\T_i\T_j}^2 l_t(\T) \right|^{2n}\les 1+e_t^{2n}+e_t^{4n} \quad \mbox{and}\quad
\left| \pa_{\T_i\T_j\T_k}^3 l_t(\T) \right|^{2n}\les 1+e_t^{2n}+e_t^{4n}.\]
Recalling that $e_t$ follows a normal distribution under $H_0$, we have
\begin{eqnarray*}
	\E\sup_{\T \in \Theta}\left| \pa_{\T_i}l_t(\T) \right|^{2n}<\infty, \quad \E\sup_{\T \in \Theta}\left| \pa_{\T_i\T_j}^2 l_t(\T) \right|^{2n}<\infty, \quad \E\sup_{\T \in \Theta}\left| \pa_{\T_i\T_j\T_k}^3 l_t(\T) \right|^{2n}<\infty, 
\end{eqnarray*}
which together with Lyapunov's inequality asserts the lemma. 
\end{proof}


\noindent{\bf Acknowledgements}\\
This research was supported by Basic Science Research Program through the National Research Foundation of Korea (NRF) funded by the Ministry of Education (NRF-2019R1I1A3A01056924).

\normalem
\catcode`'=9
\catcode``=9
\bibliography{jun}

\end{document}